\documentclass[10pt,twocolumn,twoside]{IEEEtran}
\usepackage{times}
\usepackage{cite}
\usepackage{color}
\usepackage{epsf}
\usepackage{epsfig}
\usepackage{graphicx}
\usepackage{graphics}
\usepackage{amsmath}
\usepackage{amssymb}
\usepackage{amsthm}
\usepackage{amsxtra}
\usepackage{algorithm}
\usepackage{algorithmic}
\usepackage{enumerate}
\usepackage{multirow}
\usepackage{enumerate}
\usepackage{amssymb}
\usepackage{lipsum}
\usepackage{setspace}
\usepackage{multirow}

\usepackage{mdwmath}
\usepackage{mdwtab}
\usepackage{eqparbox}
\usepackage{wrapfig}
\newtheorem{Theorem}{Theorem}

\usepackage{rotating} 
\usepackage{graphicx} 

\usepackage{url}

\newcommand{\Rmnum}[1]{\expandafter\@slowromancap\romannumeral #1@}

\makeatother
\begin{document}

\title{Power Control for Sum Rate Maximization on Interference Channels Under Sum Power Constraint}

\author{Naveed Ul Hassan, Chau Yuen, Shayan Saeed and Zhaoyang Zhang
\thanks{Copyright (c) 2013 IEEE. Personal use of this material is permitted. However, permission to use this material for any other purposes must be obtained from the IEEE by sending a request to pubs-permissions@ieee.org.}
\thanks{This research is partly supported by Lahore University of Management Sciences (LUMS) Research Startup Grant, Singapore University Technology and Design (No. SUTD-ZJU/RES/02/2011), National Key Basic Research Program of China (No. 2012CB316104), National Hi-Tech R\&D Program of China (No. 2014AA01A702) and Zhejiang Provincial Natural Science Foundation of China (No. LR12F01002).}
\thanks{N. U. Hassan and S. Saeed are with Department of Electrical Engineering, SBASSE, LUMS, Lahore, 54792, Pakistan (e-mail: naveed.hassan@yahoo.com/naveed.hassan@lums.edu.pk, shayansaeed93@gmail.com).}
\thanks{C. Yuen is with Engineering Product Development Department, Singapore University of Technology and Design, 138682, Singapore (e-mail: yuenchau@sutd.edu.sg).}
\thanks{Z. Zhang is with the Department of Information Science and Electronic Engineering, Zhejiang University, Hangzhou 310027, China (e-mail: ning\_ming@zju.edu.cn).}

}

\IEEEoverridecommandlockouts
\maketitle

\begin{abstract}
In this paper, we consider the problem of power control for
sum rate maximization on multiple interfering links (TX-RX pairs)
under sum power constraint. We consider a single frequency
network, where all pairs are operating in same frequency band,
thereby creating interference for each other. We study the power
allocation problem for sum rate maximization with and without QoS
requirements on individual links. When the objective is only sum
rate maximization without QoS guarantees, we develop an analytic 
solution to decide optimal power allocation for two TX-RX pair problem. 
We also develop a low complexity iterative algorithm for three TX-RX pair problem. 
For a generic $N>3$ TX-RX pair problem, we develop two
low-complexity sub-optimal power allocation algorithms. The first
algorithm is based on the idea of making clusters of two
or three TX-RX pairs and then leverage the power allocation results obtained 
for two and three TX-RX pair problems. The second algorithm 
is developed by using a high SINR approximation and this algorithm can 
also be implemented in a distributed manner by individual TXs. 
We then consider the same problem but with additional QoS
guarantees for individual links. We again develop an 
analytic solution for two TX-RX pair problem, and a distributed
algorithm for $N > 2$ TX-RX pairs.

\end{abstract}

\section{Introduction}
\label{sec:int} The aggressive re-use of wireless spectrum (due to
spectrum scarcity) in a wireless network can result in several
interfering links and a significant degradation in system
throughput. Some of the adverse effects of interference can be
mitigated through power control and cooperation among transmitters.
Power control in different types of wireless networks has been an
area of active research for last several years. Recently cooperative
communication techniques have also gained a lot of research interest
since cooperating transmitters can better manage interference and
improve the performance of a wireless network.

With an emerging interest in renewable energy sources and energy
harvesting schemes, several new and interesting problems arise in
terms of power control in certain wireless and cognitive radio networks. 
One can imagine a scenario where a battery harvested energy from solar panel,
and is providing power to multiple transmitters (sensor nodes,
distributed antennas, transmitters in small cells etc). In this
situation, the harvested energy in any time slot is shared among
multiple transmitters and the resulting optimization problem
consists of a sum power constraint on multiple interfering links.
The sum power constraint could also provide a fair comparison under
certain scenarios in heterogeneous network. This constraint may also arise 
in games played by resource-constrained players: e.g. in cognitive radio 
networks and wireless networks. These games are characterized by a central 
feature that each user has a multi-dimensional action space,
subject to a single sum resource constraint. 
Consideration of power sharing among multiple TXs is also motivated
in the emerging scenario of distributed antenna systems (DAS). In
DAS, multiple antennas are geographically placed at various
locations in the cell. These antennas are connected to a central
common source via wired connections \cite{disant_1}. Optimal power allocation
among multiple interfering transmitters for
sum rate maximization in general is a very challenging problem due
to the fact that the capacity region of an interference channel has
still not been completely characterized for two or more interfering
links \cite{int_chan}.

\subsection{State of the art}
A good review of power control techniques in wireless
networks can be found in \cite{chiang,rev_n,rev_ref1,rev_ref2}. 
Several authors have formulated optimization problems for power control in
different wireless network settings where interference is not
considered or it is treated as part of noise
\cite{non_coop1}-\cite{non_coopN}. In most of these papers, the
developed optimization problem is converted into a convex
optimization problem with zero duality gap. Lagrange dual
decomposition techniques are then used to obtain optimal power
control which results in water-filling over the inverse of channel gain values
\cite{boyd_1}. However, in a practical network where various
transmitters are operating in the same frequency band, ignoring
interference can be a huge disadvantage in terms of system
throughput. Unfortunately, by considering the influence of 
interference, the resulting power control problem becomes a
non-convex optimization problem. Most of existing results show 
that the power control problems over interfering links are usually 
NP-hard problems \cite{rev_ref6}. 

The problem of sum rate maximization for a binary interference 
link (link and TX-RX pair are used interchangeably throughout the paper) has been considered in 
\cite{int_cont1}-\cite{int_cont4} under various assumptions (\textit{binary interference link} 
means that ``TX1 is connected to RX1 and TX2 is connected to RX2 on 
the same frequency band"). In \cite{int_cont1}, \cite{int_cont2}, this problem has been studied
for the case of strong interference where the message from the
strong interferer is decoded first.
In \cite{int_cont4}, the authors consider the sum rate maximization
problem in a symmetric network of interfering links by identifying
an underlying convex structure. The authors also allow multiple
receivers to coexist in same frequency band.
In \cite{kiani_1}-\cite{kiani_n}, the authors prove that for binary
interfering wireless link, when each link has its own maximum
transmit power constraint, the maximum sum rates solution is one of three points: one link 
transmitting with full power while the other link is silent, or both links 
transmitting at full power simultaneously. These results are achieved through
the analysis of the objective function. In \cite{rev_ref10}, the authors 
analyze the rate region frontiers of $N$ user interference channel while 
treating interference as noise. They show that the achievable rate region
is the convex hull of a union of $N$ rate regions. Each region is outer-bounded
by a hyper surface frontier of dimension $N-1$. Using this analysis, the authors
in \cite{rev_ref11}, show that for binary case, the achievable rate region is 
a union of two regions and each region is outer-bounded by a log defined line.
The authors then analyze the first derivatives of the rate region frontiers
and proves the same power control results which are reported 
in \cite{kiani_1}-\cite{kiani_n} for binary interfering links. 
In \cite{kk_x}, the authors consider the power control problem 
of two interfering wireless links with individual
max power constraints and minimum data rate constraints. The authors
prove that in this problem, the optimal solution is again when both
links transmit with max power, or when one link transmits with max
power and the second transmits with sufficient power which allows it
to satisfy its minimum rate constraints. In \cite{rev_ref}, the
authors consider the sum rate maximization problem for a set of
TX-RX pairs operating in a common spectrum band. They optimize the
allocation of power spectral densities in Gaussian interference
network with flat fading and propose an iterative algorithm for
bandwidth and power allocation for multiple users. In
\cite{our_globecom}, the authors consider optimal power control and
optimal antenna selection problem in a multi-user distributed
antenna system (only two antennas) connected to a single RF chain
under sum-power constraint. The authors develop a low complexity
algorithm while ensuring the QoS requirements of delay sensitive users.

There are several works in the literature that also discuss interference management for multiple interfering links. Some popular algorithms in the literature include the interference pricing algorithm \cite{rev_ref3}, \cite{rev_ref4} and Weighted Minimum Means Square Error (WMMSE) algorithm \cite{rev_ref5}. In interference pricing algorithms \cite{rev_ref3}, \cite{rev_ref4}, each link announces an interference price that represents the marginal cost of interference from other links in the network. The links then iteratively update their power and converge to a stationary point, under certain conditions on the utility function (which is being maximized). The set of utility functions unfortunately does not include the standard Shannon rate function. The WMMSE algorithm proposed in \cite{rev_ref5}, transforms the weighted sum-rate maximization problem in MIMO broadcast downlink channel to an equivalent weighted sum MSE minimization problem. The proposed algorithm can handle fairly general utility functions and the sequences of iterates produced by the algorithm converges to a local optima with low complexity. 
Power control problem with linearly coupled constraint (sum power constraint, or more generally interference temperature constraint), is also considered in some recent works \cite{rev_ref8,rev_ref9,linear_coup}. The authors in \cite{rev_ref8} consider a cognitive radio setup, and their objective is to limit the aggregated power of the interference generated by the unlicensed network falls below a certain threshold. They propose a distributed algorithm that converges to a specified set of equilibrium. Using a game theoretical approach, the authors in \cite{rev_ref9} develop algorithms for MIMO cognitive radio network. The proposed algorithms are designed such that they do not violate the interference temperature constraints. In \cite{linear_coup}, the authors use game theoretic tools to analyze a broad family of games played by resource-constrained players. These games are characterized by a central feature that each user has a multi-dimensional action space, subject
to a single sum resource constraint and each user's utility in a particular dimension depends on an additive
coupling between the user action in the same dimension and the actions of the other users.
The authors then explore the properties of these games and provide several sufficient
conditions under which best response dynamics converges linearly to the unique NE.

In general, optimal power control for sum rate maximization under individual maximum power constraint per interfering link or sum power constraint for multiple interfering links (with or without QoS guarantees for individual links) remains unknown.

\subsection{Problem Addressed and motivation}
In this paper, we consider the power control problem for sum rate maximization over multiple interfering TX-RX pairs connected to a common energy source. Each RX is assumed to be connected to only one TX. 
The consideration that all the TXs are connected to a common energy source gives rise to a sum power constraint i.e. there is a constraint on total power budget available for multiple transmitters denoted by $P_T$. Moreover, since each TX has access to all the available supply, therefore its transmit power constraint is also equal to $P_T$ and the power allocated to any TX $i$ to communicate with its corresponding RX can be any value between $0$ and $P_T$.
In centralized solution, all the TXs are connected to a centralized
node / controller. Power allocation decisions are carried out at the
central node based on channel feedback information obtained from
various TXs. Power allocation decisions are then communicated to
individual TXs which then draw the allocated amount of power from
the common energy source. In some cases, we also consider
distributed solution, where each TX individually decides its
transmit power, and centralized node is not required. 
Consideration of sum power constraint
is mainly motivated in following emerging scenarios: 
\begin{itemize}
\item Energy harvesting based communication systems: an energy harvester captures green
energy from the environment. The harvested energy is used to operate
multiple transmitters, e.g. solar panel powered battery which
powered multiple transmitters.
\item Distributed antenna systems: In DAS, multiple antennas are geographically
placed at various locations. These antennas are connected to a common central source via wired connections.
\item For fair comparison in heterogeneous networks: e.g. under same sum power constraint, some macro BS
traffic being offload to femto BS.
\item Games played by resource-constrained players: e.g. in cognitive radio networks and wireless networks.
These games are characterized by a central feature that each user has a multi-dimensional action space,
subject to a single sum resource constraint.
\end{itemize}
Furthermore, the solution of sum rate maximization problem under sum power constraint will also provide an upper bound under cooperation in power sharing and power transfer which might become possible in future for some communication systems over short distances.

\subsection{Contributions}
The major contributions of this work are listed below:
\begin{itemize}
    \item \textbf{Optimal power control algorithm for two TX-RX pairs without QoS guarantees}: We develop an optimal power control for sum rate maximization under sum power constraint for two TX-RX pair problem. The developed solution consists of a very simple criterion to decide between binary or sharing type of power controls. In binary power control, only one TX gets full power while the second TX remains silent. In sharing power control, both the transmitters get non-zero power. In this case, the optimal power allocation to each TX can be determined by solving a simple quadratic equation. 
  \item \textbf{Low complexity power control algorithm for three TX-RX pairs without QoS guarantees}: We develop a low complexity iterative power control algorithm for three TX-RX pair problem. In this algorithm, transmit power of one TX is iteratively updated while the transmit power of remaining two TXs is analytically determined. 
  \item \textbf{Low Complexity sub-optimal power control algorithms for $N>3$ TX-RX pairs without QoS guarantees}: We develop two low complexity sub-optimal algorithms for $N>3$ TX-RX pair problem. The first algorithm is based on clustering technique. The second algorithm is developed by using a high SINR approximation and this algorithm can also be implemented in a distributed manner by individual TXs without the requirement of the central node. The performance evaluation based on the simulated network scenarios in section IV indicate that the clustering algorithm performs well when total available power is less while the distributed algorithm performs well when total available power is high.  
    \item \textbf{Analytical power control algorithm for two TX-RX pairs with QoS guarantees for individual links}: We develop an analytical power control algorithm for sum rate maximization when there are additional requirements of providing QoS guarantees for individual links. The developed solution in this case again comprises of solving simple equations.
    \item \textbf{Low Complexity sub-optimal distributed power control algorithm for $N > 2$ TX-RX pairs with QoS guarantees for individual links}: Using a high SINR approximation, we develop a low complexity sub-optimal power control algorithm. This algorithm provides QoS guarantees for individual links and can be implemented in a distributed manner for any number of TX-RX pairs in the network.
\end{itemize}

\subsection{System Model}
In this paper, we consider a generic system model consisting of
multiple TX-RX pairs located inside a certain geographical area. The
TXs are connected (on a wired channel) to a central node. 
Each RX is assumed to estimate its direct channel
gain from its own TX and interference channel gains from other
interfering TXs. This information is then communicated to the
central node. We assume that the channel gains among TXs and RXs
remain constant for the duration of a given time slot but may vary
from one time slot to another. The central node as well as all the
TXs in the system are powered by a common energy source, e.g. a
common battery charged by a solar panel. Power control decisions are
made at the central node and these decisions are communicated to
corresponding TXs. Based on power allocation decisions, each TX
draws the allocated amount of power from the energy source.
The system model is shown in Fig. \ref{fig:sys}.
\begin{figure}[htb]
\centering
\includegraphics[width=.4\textwidth,height=.22\textheight]{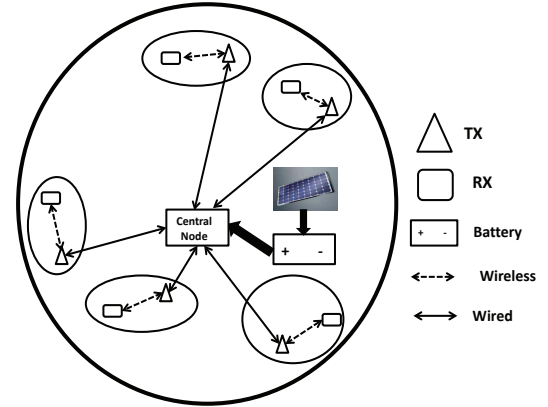}
\caption{System Model}
\label{fig:sys}
\end{figure}
In the paper we also develop some distributed algorithms. It should
be noted that the distributed algorithms can be implemented by
individual TXs without the requirement of the central node.
We assume that all the TXs are operating in same frequency band and
creating interference to each other. The interference channel for
five interfering links is depicted in Fig. \ref{fig:nint}.
\begin{figure}[htb]
\centering
\includegraphics[width=.48\textwidth,height=.18\textheight]{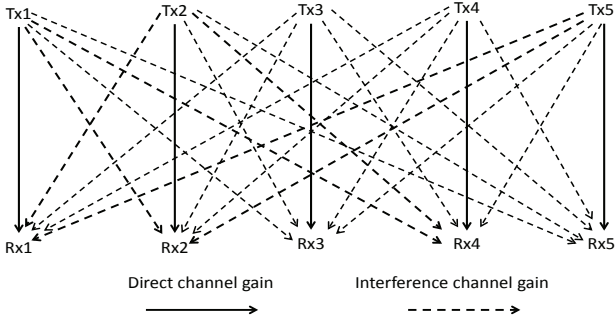}
\caption{Interference Channel: Five TX-RX pairs} \label{fig:nint}
\end{figure}

\subsection{Paper Organization}
The rest of the paper is organized as follows. In Section
\ref{sec:noqos} we formulate the optimization problem without QoS
constraints and propose power control algorithms for two, three and
$N>3$ TX-RX pairs. In Section \ref{sec:qos} we formulate the
optimization problem with QoS constraints for individual links. In
this section, we present power control algorithms for two and $N>2$
TX-RX pairs. Simulation results are presented in Section
\ref{sec:sim} while the paper is concluded in Section
\ref{sec:conc}.

\section{Power Control Without QoS Constraints}
\label{sec:noqos}
In this section, we formulate the optimization problem and develop power control algorithms when there is no QoS guarantee for individual links.
Let $P_i$ denote the power allocated to TX $i$ to communicate with its respective RX $i$. Our objective is to make power allocation decisions for sum-rate maximization subject to sum power constraint. We can formulate the following optimization problem:
\begin{equation}
\max_{P_1,P_2,...,P_N} \:  \mathcal{C}(P_1,P_2,...,P_N)
\label{global_opt1}
\end{equation}
subject to,
\begin{equation}
\sum_{i=1}^N P_{i} \leq P_T
\label{consty1}
\end{equation}
\begin{equation}
0\leq P_i \leq P_T \quad \forall i=1,\ldots, N
\label{consty2}
\end{equation}
where $\mathcal{C}(P_1,P_2,...,P_N)$ represents sum-rate function.
Constraint (\ref{consty1}) is the sum power constraint.
This constraint is a linear coupling constraint which indicates that
the total power resource is constrained. 
Constraint (\ref{consty2}) are individual power constraints per TX.
We assume that each TX can adaptively vary its transmit power depending 
on the power allocation decisions and the power amplifier is capable to 
transmit with with any value of $P_i$ between $0$ and $P_T$.
We can use the widely used Shannon's capacity formula to define the sum-rate function as
follows:
\begin{align}
\mathcal{C}(P_1,P_2,...,P_N) &= \sum_{i=1}^N R_i(P_1,P_2,...P_N) \\ &=\sum_{i=1}^N \log_2(1+ \frac{g_{ii} P_i}{\sum_{j \neq i}P_jg_{ji}+\sigma^2})
\label{global_opt2}
\end{align}
In this formulation, $g_{ji}$ denotes the channel gain from TX $j$ to RX $i$ while $\sigma^2$ is the Additive White Gaussian Noise (AWGN) noise and $\sum_{j \neq i}P_jg_{ji}$ is the interference experienced by TX-RX pair $i$ from all other TXs. The objective function (\ref{global_opt2}) is a non-convex function as can be verified from its Hessian. Due to non-convexity of the problem, it is a difficult problem to solve since convex optimization techniques cannot be applied. Convex optimization has several advantages compared to non-convex 
or nonlinear optimization techniques. For convex optimization problems; 
duality gap is zero and any local optimum is also a global optimum of the problem; 
efficient algorithms can be easily developed \cite{boyd_1}. 
For non-convex or non-linear optimization problems, meta-heuristics e.g. neighborhood search
algorithms, simulated annealing algorithms, genetic algorithms etc. are used.
These algorithms converge slowly and can even fail to return a global optimum 
if fast convergence is required (since local optimum and global optimum
are not the same).

\subsection{Optimal Power Control for two TX-RX pair problem without QoS guarantees}
\label{sec:noqos2}
The optimization problem for two TX-RX pairs without QoS guarantees for individual links can be written as,
\begin{equation}
\max_{P_1, P_2}\: \mathcal{C}(P_1,P_2) \label{my_prob2}
\end{equation}
\begin{equation}
\text{subject to,} \quad P_1+P_2 \leq P_T\label{my_const2}
\end{equation}
\begin{equation}
0\leq P_1 \leq P_T, \quad \quad 0 \leq P_2 \leq P_T
\end{equation}
\begin{align}
\text{where,}\quad & \nonumber \mathcal{C}(P_1,P_2)  \\ & \nonumber = \log_2(1+ \frac{g_{11} P_1} {g_{21} P_2+\sigma^2}) + \log_2(1+
\frac{g_{22} P_2} {g_{12} P_1+\sigma^2} )
\\ &=\log_2(1+ \frac{a P_1} {b P_2+1}) +\log_2(1+ \frac{d P_2} {c
P_1+1})
\label{my_prob21}
\end{align}
Here $a$, $b$, $c$, $d$ are four positive constants:
$a=g_{11}/\sigma^2$, $d=g_{22}/\sigma^2 $, $b=g_{21}/\sigma^2 $ and
$c=g_{12}/\sigma^2 $, where $a$ and $d$ denote the direct channel
gains while $b$ and $c$ denote the interference channel gains.
It is obvious that the optimal solution to the above optimization problem
(\ref{my_prob2}) can result in either binary power control, i.e. $(P_T,0)$ or $(0,P_T)$ 
or sharing power control, i.e. $(0<P_1 <P_T , 0<P_2<P_T-P_1)$.
We will now find the conditions under which these solutions are optimal without imposing any restrictions on the channel gain values. 
Depending on the direct channel gain values, there are two possibilities: $a>d$ or $d>a$. 

\subsubsection{Case: $a > d$}
In this case the direct channel gain of TX-RX pair 1 is greater than the direct channel gain of TX-RX pair 2. As a consequence, if binary power allocation scheme turns out to be optimal then all the power should be assigned to TX-RX pair 1. In this case, we can decide the type of power control (binary or power sharing) based on the following simple rule.
\begin{Theorem}\label{them_1}
Compute $\phi=\frac{(a-d)(P_Tb+1)}{ac+bd-ad+P_Tabc}$; then depending on the value of $\phi$ decide the type of power control as follows:
\begin{enumerate}
\item If $\phi >0$, $(P_T,0)$ is the optimal solution.
\item If $\phi < 0$ and $|\phi| \geq P_T$, again $(P_T,0)$ is the optimal solution.
\item If $\phi < 0$ and $|\phi| \leq P_T$, there exists a power sharing profile $(0<P_1 <P_T , 0<P_2 <P_T-P_1)$ such that $R_1(P_T,0)<R_1(P_1,P_2)+R_2(P_1,P_2)$.
\end{enumerate}
$|x|$ denotes the absolute value of $x$.
\end{Theorem}
\begin{proof}
Proof is provided in the conference version of this paper in \cite{our_spawc}. 
\end{proof}
\subsubsection{Case: $d > a$}
In this case the direct channel gain of TX-RX pair 2 is greater than the direct channel gain of TX-RX pair 1. We have the following theorem for this case.
\begin{Theorem}\label{them_2}
Compute $\epsilon=\frac{(d-a)(P_Tc+1)}{ac+bd-ad+P_Tbcd}$; then depending on the value of $\epsilon$ decide the type of power control as follows:
\begin{enumerate}
\item If $\epsilon >0$, $(0,P_T)$ is the optimal solution.
\item If $\epsilon < 0$ and $|\epsilon| \geq P_T$, again $(0,P_T)$ is the optimal solution.
\item If $\epsilon < 0$ and $|\epsilon| \leq P_T$, there exists a power sharing profile $(0<P_1 <P_T , 0<P_2 <P_T)$ such that $R_2(0,P_T)<R_1(P_1,P_2)+R_2(P_1,P_2)$.
\end{enumerate}
\end{Theorem}
\begin{proof}
Similar to the proof of Theorem I. 
\end{proof}
\subsubsection{Optimal values of $P_1$ and $P_2$ for power sharing solution}
For power sharing solution optimal values of of $P_1$ and $P_2$ have to be determined.
We have the following theorem to determine these values for $a>d$ case (the results for $d>a$ can be obtained in a similar way).
\begin{Theorem}\label{them_3} For $a>d$ and the case of power sharing, there exists a unique optimal point $(0<P_1^*<P_T,0<P_2^*<P_T)$, which maximizes the sum rate function $\mathcal{C}(P_1^*,P_2^*)$. The optimal value of $P_1^*$ can be found by solving the following quadratic equation (out of two possible values only one value lies between $0<P_1^*<P_T$ and the second value is outside these limits and is not the solution),
\begin{equation}
AP_1^2+B P_1+C = 0
\label{solu_1}
\end{equation}
\[\text{where:} \: A= ad(b - c) + c(ac + P_Tbca) - b(bd +  P_Tbcd)\]
\[B=-2m(1+P_Tb)\]
\[C=(1+P_Tb)(a-d)+P_Tm (1+P_Tb)\]
\[m=ad-(ac+bd+P_Tbcd) \]
The value of $P_2^*$ is: $P_2^*=P_T-P_1^*$.
\end{Theorem}
\begin{proof} Proof is provided in the conference version of this paper in \cite{our_spawc}.
\end{proof}
Due to symmetry of the problem, we can obtain the optimal solution
for $d>a$ in a similar way. 

The main idea in the proof of Theorem 1 and 2 lies in directly analyzing the objective function, which is a function of single variable (due to sum power constraint), and then finding the conditions under which binary power control is better than sharing power control. In Theorem 3, we show that if the solution is sharing type of power control, then the resulting objective function (again in single variable) is convex by looking at its second derivative. Next, we find the optimal power sharing point by taking the first derivative of the objective function that results in a quadratic equation with two solutions. We show that only one solution of this quadratic equation is valid since it lies in the interval $[0,P_T]$, and the value given by second solution is either negative or greater than $P_T$. These proofs are detailed in \cite{our_spawc}. Please note that some authors \cite{rev_ref10}, \cite{rev_ref11} have derived power control results for binary interfering links (without sum power constraint) by analyzing the rates region frontiers. 
In our case, we adopt a different approach, and make use of the sum power constraint which reduces the objective function into an unconstrained optimization problem in single variable. 
We summarize power control results for two TX-RX pairs in the form of Algorithm \ref{algI}.
\begin{algorithm}[htb]
\caption{Optimal Power Control Algorithm for two TX-RX pairs without QoS guarantees}
\label{algI}
If $a>d$, determine the value of $\phi$ as defined in Theorem \ref{them_1}:
\begin{algorithmic}[1]
\STATE If $\phi$ satisfies conditions 1 or 2 outlined in Theorem \ref{them_1}, optimal power control is $(P_T,0)$.
\STATE If $\phi$ satisfies condition 3 outlined in Theorem \ref{them_1}, solve quadratic equation (\ref{solu_1}) in Theorem \ref{them_3} to determine optimal values of $(P_1^*,P_2^*)$.
\end{algorithmic}
If $d>a$, determine the value of $\epsilon$ as defined in Theorem \ref{them_2}:
\begin{algorithmic}[1]
\STATE If $\epsilon$ satisfies conditions 1 or 2 outlined in Theorem \ref{them_2}, optimal power control is $(0,P_T)$.
\STATE If $\epsilon$ satisfies condition 3 outlined in Theorem \ref{them_2}, solve a quadratic equation similar to equation (\ref{solu_1}) in Theorem \ref{them_3} to determine optimal values of $(P_1^*,P_2^*)$.
\end{algorithmic}
\end{algorithm}

\subsection{Power Control for three TX-RX pair problem without QoS guarantees}
\label{sec:noqos3}
In the optimization problem for three TX-RX pairs without QoS constraints per individual link we can write the sum-rate $\mathcal{C}(P_1,P_2,P_3)$according to the Shannon's capacity formula as,
\begin{align}
 \mathcal{C}(P_1, P_2,P_3) = &
  \log_2(1+\frac{g_{11}P_1}{g_{21}P_2+g_{31}P_3+\sigma^2}) \nonumber \\& +  \log_2(1+\frac{g_{22}P_2}{g_{12}P_1+g_{32}P_3+\sigma^2}) \nonumber \\& + \log_2(1+\frac{g_{33}P_3}{g_{13}P_1+g_{23}P_2+\sigma^2})
  \label{three_cap}
\end{align}
For a given value of $P_1$, (\ref{three_cap}) can be expressed as,
\begin{align}
 \mathcal{C}(P_2,P_3) = & \log_2(1+\frac{a_1P_2}{b_1P_3+1}) + \log_2(1+\frac{d_1P_3}{c_1P_2+1}) \nonumber \\& + \log_2(1+\frac{e_1}{f_1P_2+h_1P_3+1})
 \label{my_eq3}
\end{align}
where $a_1=\frac{g_{22}}{g_{12}P_1+\sigma^2}$ ,
$b_1=\frac{g_{32}}{g_{12}P_1+\sigma^2}$ ,
$c_1=\frac{g_{23}}{g_{13}P_1+\sigma^2}$ ,
$d_1=\frac{g_{33}}{g_{13}P_1+\sigma^2}$,
$e_1=\frac{g_{11}P_1}{\sigma^2}$, $f_1=\frac{g_{21}}{\sigma^2}$ and
$h_1=\frac{g_{32}}{\sigma^2}$ are all positive constants (for given
value of $P_1$). Finding the optimal solution to the optimization
problem for three interfering links is much harder compared to two TX-RX pair
case because the equation becomes more complicated due to the
additional interference terms from third TX. We therefore propose an iterative
Algorithm \ref{algII} based on the following theorem.
\begin{Theorem}\label{them_4} For a given value of $P_1$, the remaining power $\bar{P}=P_T-P_1$ can be allocated to TX2 and TX3. The optimal tuple ($P_2^*,P_3^*$) that maximizes $\mathcal{C}(P_2,P_3)$ in equation (\ref{my_eq3}) is either $(0,\bar{P})$ or $(\bar{P},0)$ or $P_3^*$ is one of the solution of the following quartic equation and $P_2^*=\bar{P}-P_3^*$:
\begin{equation}
A_1P_3^4+B_1P_3^3+C_1P_3^2+D_1P_3+E_1=0
\end{equation}
\begin{align}
\text{where}, \quad A_1  = & a_1^{'}b_1c_1d_1^{'}h_1^{'2} - b_1^{'}c_1d_1^{'}h_1^{'2} + b_1b_1^{'}c_1c_1^{'}h_1^{'2} + \nonumber \\ &  b_1b_1^{'}c_1^{'}d_1^{'}h_1^{'2} + b_1b_1^{'}c_1d_1^{'}e_1^{'}h_1^{'} - b_1b_1^{'}c_1d_1^{'}f_1^{'}h_1^{'} \nonumber
\end{align}
\begin{align}
B_1  = & 2b_1^{'}c_1^{'}d_1^{'}h_1^{'2} - 2b_1^{'}c_1d_1^{'}f_1^{'}h_1^{'} + 2a_1^{'}b_1c_1c_1^{'}h_1^{'2} + \nonumber \\ & 2a_1^{'}b_1c_1d_1^{'}e_1^{'}h_1^{'} + 2b_1b_1^{'}c_1c_1^{'}e_1^{'}h_1^{'} + 2b_1b_1^{'}c_1^{'}d_1^{'}f_1^{'}h_1^{'} \nonumber
\end{align}
\begin{align}
& C_1= b_1^{'}c_1^{'2}h_1^{'2} + a_1^{'}c_1c_1^{'}h_1^{'2} + a_1^{'}c_1^{'}d_1^{'}h_1^{'2} - a_1^{'}b_1c_1^{'2}h_1^{'2} - \nonumber \\ & b_1^{'}c_1d_1^{'}e_1^{'}f_1^{'} + a_1^{'}c_1d_1^{'}e_1^{'}h_1^{'} + b_1^{'}c_1c_1^{'}e_1^{'}h_1^{'} - a_1^{'}c_1d_1^{'}f_1^{'}h_1^{'} - \nonumber \\ & b_1^{'}c_1c_1^{'}f_1^{'}h_1^{'} + b_1^{'}c_1^{'}d_1^{'}e_1^{'}h_1^{'} + 3b_1^{'}c_1^{'}d_1^{'}f_1^{'}h_1^{'} - b_1b_1^{'}c_1^{'2}e_1^{'}h_1^{'} + \nonumber \\ & b_1b_1^{'}c_1^{'2}f_1^{'}h_1^{'} + a_1^{'}b_1c_1d_1^{'}e_1^{'}f_1^{'} + b_1b_1^{'}c_1c_1^{'}e_1^{'}f_1^{'} + 3a_1^{'}b_1c_1c_1^{'}e_1^{'}h_1^{'} + \nonumber \\ & b_1b_1^{'}c_1^{'}d_1^{'}e_1^{'}f_1^{'} + a_1^{'}b_1c_1c_1^{'}f_1^{'}h_1^{'} - a_1^{'}b_1c_1^{'}d_1^{'}e_1^{'}h_1^{'} + a_1^{'}b_1c_1^{'}d_1^{'}f_1^{'}h_1^{'} \nonumber 
\end{align}
\begin{align}
D_1= & 2b_1^{'}c_1^{'2}f_1^{'}h_1^{'} + 2a_1^{'}c_1c_1^{'}e_1^{'}h_1^{'} + 2b_1^{'}c_1^{'}d_1^{'}e_1^{'}f_1^{'} + \nonumber \\ & 2a_1^{'}c_1^{'}d_1^{'}f_1^{'}h_1^{'} - 2a_1^{'}b_1c_1^{'2}e_1^{'}h_1^{'} + 2a_1^{'}b_1c_1c_1^{'}e_1^{'}f_1^{'} \nonumber
\end{align}
\begin{align}
E_1= & b_1^{'}c_1^{'2}e_1^{'}f_1^{'} - a_1^{'}c_1^{'2}e_1^{'}h_1^{'} + a_1^{'}c_1^{'2}f_1^{'}h_1^{'} + \nonumber \\ & a_1^{'}c_1c_1^{'}e_1^{'}f_1^{'} + a_1^{'}c_1^{'}d_1^{'}e_1^{'}f_1^{'} - a_1^{'}b_1c_1^{'2}e_1^{'}f_1^{'} \nonumber
\end{align}
\[a_1^{'}=a_1\bar{P}+1, \: \: b_1^{'}=b_1-a_1, \: \: c_1{'}=c_1\bar{P}+1, \: \: d_1^{'}=d_1-c_1\] \[ e_1^{'}=f_1\bar{P}+e_1+1, \: \: f_1{'}=f_1\bar{P}+1, \: \: h_1^{'}=h_1-f_1  \]

\end{Theorem}
\begin{proof}
The proof is given in Appendix \ref{app_1}.
\end{proof}
\begin{algorithm}[htb]
\caption{Power Control Algorithm for three TX-RX pairs without QoS guarantees}
\label{algII}
Initialize: Select step size = $\nu >0$, iteration index $m=1$ and $M= \left\lfloor \frac{P_T}{\nu}\right\rfloor$ \\
Initialize: $P_1(0)=-\nu$ 
\begin{algorithmic}[1]
\FOR{$m=1,2,\ldots \ldots M+1$ }
\STATE $P_1(m)=P_1(m-1)+\nu$ and $\bar{P}=P_T-P_1(m)$.
\STATE Use Theorem \ref{them_4} to analytically determine optimal allocation of $\bar{P}$ among TX2 and TX3 i.e. $P_2(m)$ and $P_3(m)$. 
\STATE Find sum rate $\mathcal{C}_m(P_1(m),P_2(m),P_3(m))$ using (\ref{three_cap}) and save the values of $(P_1(m),P_2(m),P_3(m))$ and the resulting sum rate.
\ENDFOR
\STATE Select the index $m^*$ that achieved largest sum rate i.e.
\[m^*=\max_m \: \mathcal{C}_m(P_1(m),P_2(m),P_3(m)) \]
\STATE The corresponding power allocation is $(P_1^*,P_2^*,P_3^*)=(P_1(m^*),P_2(m^*),P_3(m^*))$.
\end{algorithmic} 
\end{algorithm}
In this algorithm, we select a small positive step size $\nu >0$, with an iteration on the variable $P_1$. The total number of iterations starting from an initial value of $P_1(1)=0$ are denoted by $M+1$ ($\left\lfloor x \right\rfloor$ takes the integer part after division). In each iteration $m$, for a given value of $P_1(m)$, the remaining power $\bar{P}=P_T-P_1(m)$ is allocated to TX2 and TX3. We use Theorem \ref{them_4} to analytically determine the optimal allocation of $\bar{P}$ among TX2 and TX3 and denote it by $P_2(m)$ and $P_3(m)$. In step 4, total sum rate is computed using (\ref{three_cap}) and the values of $(P_1(m),P_2(m),P_3(m))$ and the achieved sum rate in current iteration are stored. When the loop on variable $P_1$ terminates, we have a set of $M+1$ sum rate values each corresponding to a particular power allocation. In step 6, we determine the index $m^*$ which maximizes the sum rate. Finally in step 7, the power allocation scheme is obtained as $(P_1^*,P_2^*,P_3^*)=(P_1(m^*),P_2(m^*),P_3(m^*))$.

This algorithm is an exhaustive search on variable $P_1$. However, the remaining two 
variables $P_2$ and $P_3$ are determined analytically. We want to highlight 
that the complexity of this algorithm is very low as compared to a pure exhaustive search algorithm, 
where all three variables have to be found in an iterative way.
The quality of the solution provided by this algorithm depends on the value of step size $\nu$. 
In general, small value of $\nu$ will lead to a better solution
compared to a larger value. 
On the other hand, using a smaller step size also increases the number of iterations 
required to obtain a solution.

\subsection{Power control for $N>3$ TX-RX pair problem without QoS guarantees}
\label{sec:noqosn}
The optimization problem for $N>3$ TX-RX pairs without QoS constraints for individual links, is quite hard to solve. In this paper we present two low complexity sub-optimal algorithms; one performs well when total power is less and another performs well when total power is high.

\subsubsection{Clustering Algorithm for $N$ TX-RX pair problem without QoS guarantees}
\label{sec:noqosclust} For $N$ TX-RX pairs, leveraging on Algorithms 1 and 2, we now develop a clustering algorithm, which is a sub-optimal heuristic, for determining the power control. The algorithm is based on the idea of making clusters of two or three TX-RX pairs. 
Clustering Algorithm \ref{algIII} explains power allocation among $N$ TX-RX pairs.
\begin{algorithm}[htb]
\caption{Clustering Algorithm for $N>3$ TX-RX pairs without QoS guarantees}
\label{algIII}
Initialize: Select $r=2$ or $r=3$, $K=\left\lfloor \frac{N}{r}\right\rfloor$, $\hat{P}=\frac{P_T}{N}$.
\begin{algorithmic}[1]
\FOR{all cluster formation options}
\STATE Power allocated to each cluster is $\tilde{P}=r \hat{P}$.
\STATE TX-RX pairs which are not part of any cluster are allocated $\hat{P}$.
\STATE For each cluster, central node estimates interference according to \eqref{clus:int}.
\STATE Within each cluster, allocate power among TX-RX pairs using Algorithm \ref{algI} if $r=2$ or Algorithm \ref{algII} if $r=3$.
\STATE Compute the achieved sum rate. 
\ENDFOR
\STATE The cluster formation which achieves highest sum rate is selected.
\end{algorithmic}
\end{algorithm}

Given $N$ TX-RX pairs, we form $K=\left\lfloor \frac{N}{r}\right\rfloor$ clusters where each cluster comprises of $r$ TX-RX pairs ($r\in \{2,3\}$). There are numerous ways to form clusters. Any $N$ TX-RX pairs can be arranged in $\binom{N}{r}=\frac{N!}{r!(N-r)!}$ possible ways. For a moderate number of TX-RX pairs e.g. $N=10$, taking $r=2$ results in $45$ groups while taking $r=3$ results in $120$ groups. Optimization has to be done over all possible cluster formation options. In clustering algorithm, in order to use the derived results for 2 and 3 TX-RX pairs, we impose a total power constraint on each cluster. Finding an optimal power allocation for each cluster in itself is a complicated problem. Therefore, we allocate a fixed total power constraint for each cluster by equally dividing the total available power, which is the simplest and easiest way to introduce a power constraint per cluster. Let $P_c=r \hat{P}$ denote the power allocated to each cluster, where $\hat{P}=\frac{P_T}{N}$. In each cluster, $P_c$ power has to be allocated to $r$ TX-RX pairs of this cluster. This power allocation can be determined using Algorithm 1 if $r=2$ or Algorithm 2 if $r=3$. In the clustering algorithm, central node estimates interference. Central node can accurately estimate interference if it knows all the channel gain values as well as the transmit powers. Channel gain values (direct as well as cross) are assumed to be known at the central node. Central node also knows the total power allocated to each clusters since it is assumed that power is equally divided among various clusters. However, transmit powers are yet to be allocated to individual TXs in each cluster; therefore, in estimating interference for TX-RX pair $i$ (assume that this TX-RX pair belongs to cluster $c$), central node assumes that all other $N-r$ TXs in the network (excluding $r-1$ TXs which are in the same cluster $c$ as the $i$-th TX) are transmitting with $\hat{P}$ amount of power. Then the interference created for TX-RX pair $i$ is given as,
\begin{equation}
\mathcal{I}_i=\sum_{j \notin c} \hat{P} g_{ji}
\label{clus:int}
\end{equation}
This interference is treated as noise. The interference from remaining $r-1$ TXs that are in the same cluster $c$ depend on the power allocation among $r$ TXs in this cluster, i.e. power allocated to this cluster is allocated among these TXs such that interference is minimized and sum-rate of the cluster gets maximized. Among all possible cluster formation options, the one which achieves highest sum rate is selected.

In general, interference can be managed in a better way if power is optimally allocated to a cluster comprising of three TX-RX pairs as compared to clusters comprising of two TX-RX pairs. However, Algorithm 2 which is designed for three TX-RX pairs is an iterative algorithm with higher complexity as compared to Algorithm 1. Furthermore, since optimization has to be done over all possible cluster formation options, therefore the complexity of clustering algorithm with $r=3$ is higher than $r=2$. 
Clustering algorithm is a sub-optimal heuristic designed to utilize the analytical results presented in Algorithms 1 and 2. The sub-optimality of this algorithm stems from many simplifying assumptions e.g. equal power allocation among clusters, estimation of interference from remaining clusters assuming equal power allocation for interfering TXs etc. When total power is low and $N$ is large, power assigned to each TX is low. In this case, the impact of allocated power on resulting interference is also low as compared to channel gain values and interference estimation is more accurate. On the other hand, when total power is high, relatively more power is assigned to individual TXs and hence interference estimation tends to be less accurate. The performance of clustering algorithm therefore is better when total available power is low as compared to the case when total available power is high (we will verify this at simulation results later).
In the next subsection we develop a distributed and more practical power control algorithm. 

\subsubsection{Distributed Power Control Algorithm for $N>3$ TX-RX pairs without QoS guarantees}
\label{sec:noqosdist} In the high SINR regime, we can approximate
$\log_2(1+x)\approx \log_2(x)$. By using this approximation and a
change of variable technique, we can convert the non-convex
optimization problem (\ref{global_opt1}) into a convex optimization
problem. Using the high SINR approximation we have,
\begin{align}
\mathcal{C}(P_1,P_2,...,P_N)&=\sum_{i=1}^N \log_2(\frac{g_{ii} P_i}{\sum_{j \neq i}P_jg_{ji}+\sigma^2}) \nonumber \\& =
\frac{1}{\ln(2)}\sum_{i=1}^N \ln(\frac{g_{ii} P_i}{\sum_{j \neq i}P_jg_{ji}+\sigma^2})
\label{high_opt1}
\end{align}
where $\ln$ denotes the natural logarithm and
$\log_2(x)=\frac{\ln(x)}{\ln(2)}$ (the use of natural logarithm is
convenient in further derivations). We now define, $P_{i}^{'} =
\ln(P_i)$ and write (\ref{global_opt1}) as,
\begin{equation}
\max_{P_{1}^{'},P_{2}^{'},...,P_{N}^{'}} \:
\frac{1}{\ln(2)} \sum_{i=1}^N \ln(\frac{g_{ii} e^{P_{i}^{'}}}{\sum_{j \neq i}e^{P_{j}^{'}}g_{ji}+\sigma^2})
\label{new_prob}
\end{equation}
\begin{equation}
\text{subject to,}\: \sum_{i=1}^N e^{P_{i}^{'}} \leq P_T
\label{new_constu}
\end{equation}
With this change of variable we can verify that the Hessian matrix for the objective function (\ref{new_prob}) is a strictly negative definite matrix in the new optimization variable $P_{i}^{'}$. Hence the objective function is strictly concave in the optimization variables. We can solve this convex optimization problem using the Lagrange dual decomposition theory. Let $\lambda$ be the Lagrange multiplier associated with constraint (\ref{new_constu}). We can define the following Lagrangian function,
\begin{align}
\mathcal{L}_1(P_{1}^{'},P_{2}^{'},...,P_{N}^{'},\lambda) =  \frac{1}{\ln(2)} & \sum_{i=1}^N \ln(\frac{g_{ii} e^{P_{i}^{'}}}{\sum_{j \neq i}e^{P_{j}^{'}}g_{ji}+\sigma^2}) + \nonumber \\ & \lambda(P_T - \sum_{i=1}^N e^{P_{i}^{'}})
\label{lagrange}
\end{align}
Since the problem is convex, Karush-Kuhn-Tucker (KKT) conditions are sufficient to obtain the solution. For any TX-RX pair $i$, we put,
\[\frac{\partial \mathcal{L}_1(P_{1}^{'},P_{2}^{'},...,P_{N}^{'},\lambda)}{\partial P_{i}^{'}} = 0\]
which leads us to,
\begin{equation}
1 - \sum_{k \ne i}\frac{g_{ik} e^{P_{i}^{'}}}{\sum_{j \neq k}e^{P_{j}^{'}}g_{jk}+\sigma^2} - \lambda \ln(2) e^{P_{i}^{'}} = 0
\label{lagu_1}
\end{equation}
Now we switch back to the original optimization variables by putting $e^{P_{i}^{'}} = P_i$ in (\ref{lagu_1}) to get,
\begin{align}
P_i & = \frac{1}{\lambda \ln(2) + \sum_{k \ne i}\frac{g_{ik} }{\sum_{j \neq k}P_j g_{jk}+\sigma^2}} \nonumber \\&
= \frac{1}{\lambda \ln(2) + \sum_{k \ne i} \frac{ g_{ik}}  {\theta_k +\sigma^2} }
\label{P_update}
\end{align}
\begin{equation}
\text{where:}\quad \theta_k = \sum_{j \neq k}P_j g_{jk}
\label{theta}
\end{equation}
If we analyze (\ref{theta}), we can see that it is the total interference experienced by RX $k$ (which is connected to TX $k$) from remaining TXs in the network. Each RX $k$ can therefore easily estimate this interference since it depends only on the information that is locally available to it ($g_{jk}$ denotes channel gain between TX $j$ and RX $k$). Furthermore, the right side of \eqref{P_update} is a ``standard interference function'' \cite{yates} (the proof is omitted because it is quite straightforward and can be easily found in many papers e.g. \cite{std_1}). Therefore, we can develop an iterative algorithm to determine the appropriate power level for each TX. If we stack the power levels attained in $m$-th iteration in a vector $\mathbb{P}^m=[P_1^m,\ldots,P_N^m]$ and represent the right side of \eqref{P_update} by a function $\mathcal{F}(\mathbb{P}^m)$, then power level in $m+1$-th iteration can be determined using,
\[P_i^{m+1}=\mathcal{F}(\mathbb{P}^m) \]
We can develop an algorithm that can be implemented in a distributed way (with some signaling overhead) at each TX without the requirement of the central node.
The value of Lagrange multiplier $\lambda$ is iteratively updated according to the sub-gradient update method, 
\begin{equation}
\lambda^{m+1} = \lambda^{m} - v(m)(P_T - \sum_{i=1}^N P_{i})
\label{lambda_update}
\end{equation}
where, $v(m)$ denotes the step size used in $m$-th iteration. Based
on these results we develop Algorithm \ref{algIV}. This algorithm is
implemented by each TX in the network. At the start of this
algorithm, each TX $i$ transmits $N-1$ interference channel gain
values $g_{ki}$ to remaining TXs in the network. It also receives
$N-1$ interference channel gain values $g_{ik}$ from other TXs. In
step 1, TX $i$ computes the value of its $\theta_i$ and broadcast it
to other TXs. In step 2, TX $i$ receives the values of $\theta_k$
from remaining TXs which enables it to compute its power level in
step 3. The computed value of power is also broadcast to remaining
TXs in the network. In step 4, TX $i$ receives the power levels
allocated by other TXs in the network. In step 5, TX $i$ computes
the absolute difference between total available power and sum power
utilized by all the TXs (denoted by $\Delta_{P}$). In step 6, TX $i$
compares the value of $\Delta_{P}$ with a small positive constant
$\delta>0$ ($\delta$ is defined to control speed of convergence). If
$\Delta_{P}>\delta$, it means that sum power constraint is not
satisfied. In this case, TX $i$ updates the value of Lagrange
multiplier and repeats the algorithm. However, if
$\Delta_{P}<\delta$, then sum power constraint is satisfied (the
difference between total power and sum power allocated to all the
TXs is less than or equal to $\delta<1$) and the algorithm stops.
\begin{algorithm}[htb]
\caption{Distributed Power Control Algorithm for $N>3$ TX-RX pairs without QoS guarantees}
\label{algIV}
Initialize $P_i=0, \: \forall i$ and $\lambda^0=0$. \\
Transmit $N-1$ interference channel gain values $g_{ki}$ to other TXs in the network. Receive
$N-1$ interference channel gain values $g_{ik}$ from remaining TXs.
\begin{algorithmic}[1]
\STATE Compute and broadcast the values of $\theta_k$ for remaining TXs.
\STATE Receive the values of of $\theta_k$ from remaining TXs ($k \neq i$).
\STATE Calculate $P_i$ using (\ref{P_update}) and broadcast it to other TXs.
\STATE Receive the values of of $P_k$ from remaining TXs ($k \neq i$).
\STATE Compute $\Delta_{P}=|P_T - \sum_{i=1}^N P_{i}|$.
\STATE If $\Delta_{P}>\delta$, update the value of Lagrange multiplier using \eqref{lambda_update} and go to step 1.
\STATE If $\Delta_{P}<\delta$ declare convergence and stop the algorithm.
\end{algorithmic}
\end{algorithm}

If $\mathcal{M}_1$ denote the number of iterations required for convergence of this algorithm then the signaling overhead $\hat{S}$ for each TX is,
$\hat{S}=(N-1)+2 \mathcal{M}_1$, where $N-1$ interference channel gain information is transmitted at the start of this algorithm and after that two values ($\theta_i, P_i$) are transmitted in each iteration. The total signaling overhead for $N$ TXs is $N \hat{S}$. Since the power update equation \eqref{P_update} is a standard interference function and the problem is approximated to a convex function, the sub-gradient iterations provide fast convergence to the optimal value of the Lagrange multiplier. The worst case complexity of sub-gradient update method is, $O(\frac{1}{\delta^2})$ \cite{sub_complex1}, \cite{sub_complex2}. The value of $\delta$ can be set by system operator to trade-off complexity and performance. Moreover, the convergence to optimal values is guaranteed from any initial non-negative values e.g, $\lambda^0=0$ \cite{bertsekas}. The step size $v(m)$ in $m+1$-th iteration can be obtained according to one of the following rules (\cite{step_s_ref}, \cite{boyd_1}),
\begin{equation}
v(m)=\frac{\zeta}{\sqrt{m}}
\label{step_rule}
\end{equation}
where $\zeta>0$ is a constant. One can also use other update rules such as $v(m)=\frac{\zeta}{m}$.
\subsection{Complexity Analysis of Algorithms without QoS guarantees}
In this subsection we will do a complexity analysis of Algorithms 1, 2, 3 and 4. 
\subsubsection{Algorithm 1} Algorithm 1 is an optimal power control algorithm for two TX-RX pairs. This algorithm is a non-iterative algorithm
which is developed based on the results of Theorems 1, 2 and 3. In this algorithm, the main complexity lies
in solving the quadratic equation (when there is a power sharing solution). The complexity
of this algorithm can therefore be expressed as $O(1)$. 
\subsubsection{Algorithm 2} Algorithm 2 is an iterative algorithm. In this algorithm, the value of $P_1$ is iteratively updated
and in each iteration, the other two variables $P_2$ and $P_3$ are determined by solving a quartic equation
($4^{th}$ order equation). Let $M+1$ denote the number of iterations on variable $P_1$ (these iterations depend on the step size value). The complexity of Algorithm 2 can then be expressed as, $(M+1) O(1)$. 
\subsubsection{Algorithm 3} Algorithm 3 is a clustering algorithm. Let, $K=\left\lfloor \frac{N}{r}\right\rfloor$ denote the number of clusters. Then for a given cluster formation, the total complexity for $K$ clusters is $O(K)$ when $r=2$ while the complexity is $(M+1)O(K)$ for $r=3$. If $L$ denotes the total number of possible cluster formations (which depend on the value of $N$), then the total complexity of the clustering algorithm is, $O(LK)$ when $r=2$ while the complexity is $(M+1)O(LK)$ for $r=3$. Note that, $L$ is a function of $\binom{N}{r}$, therefore the complexity of clustering algorithm is very high for large values of $N$. 
\subsubsection{Algorithm 4} Algorithm 4 is an iterative and distributed algorithm. 
We can notice that the most complex operation in this algorithm is the computation of $N-1$ values of $\theta_k$ each of which requires $N-1$ simple operations. Therefore, the complexity of each iteration of this algorithms is $O(N^2)$. Thus, the worst case complexity of our algorithm is $O((\frac{1}{\delta})^2 N^2)$, which is polynomial in variables $N$ and $\frac{1}{\delta}$.

\section{Power Control with QoS Constraints}
\label{sec:qos} In this section, we consider the power control
problem when there are individual QoS constraints for each TX-RX
pair in the network. This problem is more complicated since it has
additional constraints as compared to the problem discussed in
Section \ref{sec:noqos}. We assume that QoS constraints are in the
form of minimum target data rate constraints. Let $R_i^{min}$ denote
the target data rate constraint for TX-RX pair $i$. We can formulate
the following optimization problem for $N$ TX-RX pairs:
\begin{equation}
\max_{P_1,P_2,...,P_N} \:  \mathcal{C}(P_1,P_2,...,P_N)=\sum_{i=1}^N R_i(P_1,P_2,...P_N)
\label{global_opt1q}
\end{equation}
\begin{equation}
\text{subject to:}\quad R_i(P_1,P_2,...P_N) \geq R_i^{min} \:, \forall i=1,\ldots, N
\label{rate_constyq}
\end{equation}
\begin{equation}
\sum_{i=1}^N P_{i} \leq P_T
\label{consty1q}
\end{equation}
\begin{equation}
0\leq P_i \leq P_T \quad \forall i=1,\ldots, N
\label{consty2q}
\end{equation}
In this formulation, constraint \eqref{rate_constyq} represents the
additional QoS requirements. These constraints demand that the
achieved data rate of any TX-RX pair $i$ should be greater than its
minimum required data rate $R_i^{min}$. The objective function
\eqref{global_opt1q} and remaining constraints (\ref{consty1q}),
\eqref{consty2q} are same as in Section \ref{sec:noqos}.

\subsection{Power Control for two TX-RX pair problem with QoS guarantees}
\label{sec:qos2} In this section, we consider the problem for two
TX-RX pair case with QoS constraints for individual links. The
optimization problem can be written as follows,
\begin{equation}
\max_{P_1,P_2} \: \mathcal{C}(P_1,P_2)=R_1(P_1,P_2)+R_2(P_1,P_2)
\label{two_obj}
\end{equation}
\begin{equation}
\text{subject to:}\: R_1(P_1,P_2) \geq R_1^{min}, \quad R_2(P_1,P_2) \geq R_2^{min}
\label{two_rate}
\end{equation}
\begin{equation}
P_1+P_2 \leq P_T
\label{two_pow}
\end{equation}
\begin{equation}
0\leq P_1 \leq P_T \quad \quad 0\leq P_2 \leq P_T
\label{two_ind}
\end{equation}
If the target data rate constraints of both the users can be achieved with given total power $P_T$, the problem is feasible and the solution will exist. We therefore identify a feasibility checking method followed by the power control algorithm development.  
\subsubsection{Feasibility Checking Method}
A necessary and sufficient condition for feasibility is the non-emptiness of the feasible set. Let, $\textbf{R}_{Q}=\{R_1^{min},R_2^{min}\}$ denote the QoS constraint vector. Let, $\textbf{C}=\{R_1,R_2\}$ denote the achieved data rate vector. Let $P_S=P_1+P_2$ denote the sum power, where, $P_1$ and $P_2$ respectively denote the power allocated to TX1 and TX2. Given $\textbf{R}_{Q}$ and $P_T$, the feasible region can be defined as the set of all possible values of $(R_1,R_2)$ that can be achieved while simultaneously satisfying all the given constraints i.e.
\[\mathcal{V}(P_1,P_2,P_T,\textbf{R}_{Q})=\Big\{\textbf{C}: \textbf{C} \: {\geq} \: \textbf{R}_{Q} \: \cap \: P_S \leq P_T \Big\} \]
where, the vector inequality is component wise (i.e. $R_1 \geq R_1^{min}$ and $R_2 \geq R_2^{min}$). The feasible region can be viewed as the intersection of two sets: one defined by the total available power and denoted by $\mathcal{V}_1(P_1,P_2,P_T)$ while the second defined by the QoS constraint vector and denoted by $\mathcal{V}_2(\textbf{R}_{Q})$ such that,
\[\mathcal{V}(P_1,P_2,P_T,\textbf{R}_{Q}) =  \mathcal{V}_1(P_1,P_2,P_T) \cap \mathcal{V}_2(\textbf{R}_{Q}) \]
The set $\mathcal{V}_1(P_1,P_2,P_T)$, contains all the values of allocated powers $(P_1,P_2)$ such that $P_S \leq P_T$ (the values of $P_1$ and $P_2$ only satisfy the sum power constraint regardless of the QoS constraints) i.e.
\[\mathcal{V}_1(P_1,P_2,P_T)=\Big\{P_1, P_2: \: P_S \leq P_T \Big \} \]
Similarly, the set $\mathcal{V}_2(\textbf{R}_{Q})$ contains all the values of data rates $(R_1,R_2)$ greater than or equal to minimum target data rate constraints (the values of $R_1$ and $R_2$ only satisfy the QoS targets regardless of the sum power constraint) i.e.
\[\mathcal{V}_2(\textbf{R}_{Q})=\Big\{R_1,R_2: \: \textbf{C} \: {\geq} \: \textbf{R}_{Q} \Big \} \]
Each QoS constraint vector $\textbf{R}_{Q}$ has an associated power region denoted by $\mathcal{V}_3(\textbf{R}_{Q})$ and defined as the set of all values of sum power $P_S$ that can result in $\textbf{C} \: \geq \: \textbf{R}_{Q}$ i.e.
\[ \mathcal{V}_3(\textbf{R}_{Q})=\Big\{P_S: \textbf{C} \: \geq \: \textbf{R}_{Q} \Big\} \]
The minimum value of sum power in the power region $\mathcal{V}_3(\textbf{R}_{Q})$ can be obtained by solving the following optimization problem:\begin{equation}
\min_{P_1,P_2} \:\: P_S=P_1+P_2
\label{mar:1}
\end{equation}
\begin{equation}
\text{subject to:}\:\: R_1(P_1,P_2) = R_1^{min}, \quad  R_2(P_1,P_2) = R_2^{min}
\label{mar:2}
\end{equation}
In this problem, there is no sum power constraint and the objective is to find minimum total power which can result in $\textbf{C} \: = \: \textbf{R}_{Q}$. Let, $P_S^{min}=P_1^{min}+P_2^{min}$ denote the minimum sum power, where, $P_1^{min}$ and $P_2^{min}$ denote the minimum values of $P_1$ and $P_2$. The minimum sum power can be obtained by using the Shannon capacity formula, which gives us two equations in two unknowns,
\begin{equation}
a P_1^{min} - b \beta_1 P_2^{min} = \beta_1
\label{simeq:1}
\end{equation}
\begin{equation}
-c \beta_2 P_1^{min} + d P_2^{min} = \beta_2
\label{simeq:2}
\end{equation}
where, $\beta_i=2^{R_i^{min}}-1$. Solving these two equations simultaneously we get,
\begin{equation}
P_1^{min} = \frac{\beta_1 \left(d + b \beta_2 \right)}{ad-bc \beta_1 \beta_2}
\label{solt:1}
\end{equation}
\begin{equation}
P_2^{min} = \frac{\beta_2 \left(a + c \beta_1 \right)}{ad-bc \beta_1 \beta_2}
\label{solt:2}
\end{equation}
Adding $P_1^{min}$ and $P_2^{min}$ obtained from \eqref{solt:1} and \eqref{solt:2} gives us the value of $P_S^{min}$ that can also achieve $\textbf{C} \: = \: \textbf{R}_{Q}$ (all the QoS constraints with strict equality). If $P_S^{min}$ lies inside or on the boundary of $\mathcal{V}_1(P_1,P_2,P_T)$ then it means that $P_S^{min} \leq P_T$ and therefore the feasible region $\mathcal{V}(P_1,P_2,P_T,\textbf{R}_{Q})$ is non-empty (since by definition $P_S^{min}$ is sufficient to achieve the given QoS constraints with strict equality). On the other hand, if the intersection of $P_S^{min}$ and $\mathcal{V}_1(P_1,P_2,P_T)$ is empty then it means that $P_S^{min} > P_T$ and therefore minimum sum power required to satisfy the target QoS constraints with strict equality exceeds the total available power and hence the non-feasibility of the problem. Therefore, we have the following optimal feasibility checking method: \textbf{Compute the values of $P_1^{min}$ and $P_2^{min}$ using \eqref{solt:1} and \eqref{solt:2}. Find $P_S^{min}=P_1^{min}+P_2^{min}$. If $P_S^{min} > P_T$, declare non-feasibility, otherwise the problem is feasible.} When the problem is non-feasible it is not possible to satisfy the target QoS constraints with the given total power $P_T$ and solution does not exist. When the problem is feasible, we use Algorithm 5 to determine the power allocation. 

\subsubsection{Algorithm Development}
We now develop a fast analytical algorithm to find the power control. We find the power required by TX1 and TX2 under worst case SINR conditions (or maximum interference conditions) to achieve the target data rate constraints with strict equality. Given total power $P_T$, this power denoted by $P_i^{worst}$ can be computed by equating the rate function \eqref{my_prob21} by $R_i^{min}$, and then substituting $P_2=P_T-P_1^{worst}$ if $i=1$ or substituting $P_1=P_T-P_2^{worst}$ if $i=2$. This leads us to,
\begin{equation}
P_1^{worst}=\frac{\beta_1 (1+ b P_T)}{a+\beta_1 b}
\label{min_p1}
\end{equation}
\begin{equation}
P_2^{worst}=\frac{\beta_2 (1+ c P_T)}{d+\beta_2 c}
\label{min_p2}
\end{equation}
We define $\tilde{P} = P_T-P_1^{worst}-P_2^{worst}$ as the remaining power. It should be noted that the allocation of $\tilde{P}$ in any possible way cannot violate the QoS constraint of any TX-RX pair since we have already calculated $P_1^{worst}$ and $P_2^{worst}$ under maximum interference conditions. For example, if we allocate all the remaining power to TX1, then total allocated power to TX1 becomes $P_1=\tilde{P} + P_1^{worst}$ while power allocated to TX2 remains $P_2=P_2^{worst}$. In this case, TX1 creates the maximum interference to TX2 as planned (i.e. $P_1 = P_T - P_2^{worst}$), therefore (36) already ensures that $P_2=P_2^{worst}$ is sufficient to achieve the QoS constraint of TX2 with strict equality. On the other hand, if some portion of $\tilde{P}$ is also allocated to TX2 then its achieved rate is higher, while still achieving the QoS constraint.

Let $\tilde{P}_i$ denote the portion of remaining power ($\tilde{P}$) that should be allocated to TX $i$, then $P_1=P_1^{worst}+\tilde{P}_1$ and $P_2=P_2^{worst}+\tilde{P}_2$, and we can write the optimization
problem as,
\begin{equation}
\max_{\tilde{P}_1,\tilde{P}_2}\: R_1(\tilde{P}_1,\tilde{P}_2)+R_2(\tilde{P}_1,\tilde{P}_2)
\label{new_rate}
\end{equation}
\begin{equation}
\text{subject to:}\quad \tilde{P}_1+\tilde{P}_2 \leq \tilde{P}
\label{rem_pow}
\end{equation}
\begin{equation}
\text{where:}\: R_1(\tilde{P}_1,\tilde{P}_2)=\log_2\left(1+\frac{a \left(P_1^{worst}+\tilde{P}_1\right)}{b\left(P_2^{worst}+\tilde{P}_2\right)+1} \right)
\label{u1_rate}
\end{equation}
\begin{equation}
R_2(\tilde{P}_1,\tilde{P}_2)=\log_2\left(1+\frac{d \left(P_2^{worst}+\tilde{P}_2\right)}{c\left(P_1^{worst}+\tilde{P}_1\right)+1} \right)
\label{u2_rate}
\end{equation}
We have the following theorem:
\begin{Theorem}\label{them_33}
The tuple ($\tilde{P}_1^*,\tilde{P}_2^*$) that maximizes the objective function \eqref{new_rate} is either $(0,\tilde{P})$ or $(\tilde{P},0)$ or $\tilde{P}_1^*$ is one of the solution of the following quadratic equation and $\tilde{P}_2^*=\tilde{P}-\tilde{P}_1^*$:
\begin{equation}
A_2 \tilde{P}_1^2 + B_2 \tilde{P}_1 +C_2 = 0
\label{quad_qos}
\end{equation}
\[\text{where:}\quad A_2=a_2'bce_2'+b_2'cc_2'e_2'+b_2'bcd_2'-b_2'be_2'f_2' \]
\[B_2=2a_2'bcd_2'+2b_2'c_2'e_2'f_2' \]
\[C_2=a_2'c_2'e_2'f_2'+a_2'bd_2'f_2'+b_2'c_2'd_2'f_2'-a_2'cc_2'd_2' \]
\[a_2'=1+aP_1^{worst}+bP_2^{worst}+b\tilde{P}, \quad \quad b_2'=a-b \]
\[c_2'=1+bP_2^{worst}+b\tilde{P}, \quad \quad d_2'=1+cP_1^{worst}+dP_2^{worst}+d\tilde{P} \]
\[e_2'=c-d, \quad \quad f_2'=1+cP_1^{worst} \]
\end{Theorem}
\begin{proof}
The proof is given in Appendix \ref{app_2}.
\end{proof}
The final power allocation for TX $i$ is,
\[P_i^*=P_i^{worst}+\tilde{P}_i^*\quad, i=1,2 \]
Based on these results, we develop Algorithm \ref{alg_II}. Please note that this algorithm is a sub-optimal analytical algorithm which can be used when the problem is feasible and $\tilde{P} > 0 $. 
\begin{algorithm}[htb]
\caption{Analytical Power Control Algorithm for two TX-RX pairs with QoS guarantees}
\label{alg_II}
\begin{algorithmic}[1]
\STATE Find the values of $P_1^{worst}$ and $P_2^{worst}$ using equations \eqref{min_p1}, \eqref{min_p2}.
\STATE Determine $\tilde{P} = P_T-P_1^{worst}-P_2^{worst}$. 
\STATE If $\tilde{P} < 0 $, the algorithm fails.
\STATE If $\tilde{P} > 0 $, use theorem \ref{them_33} to obtain the set of four possible values of $\tilde{P}_1$ and $\tilde{P}_2$.
\STATE Find sum rate using \eqref{new_rate} for all candidate solutions and select $(\tilde{P}_1^*,\tilde{P}_2^*)$ which maximizes sum rate.
\STATE Finally $P_1^*=P_1^{worst}+\tilde{P}_1^*$ and $P_2^*=P_2^{worst}+\tilde{P}_2^*$.
\end{algorithmic} 
\end{algorithm}

\subsection{Distributed Power Control Algorithm for $N>2$ TX-RX pairs with QoS guarantees}
\label{sec:qosn}
We develop this algorithm using the high SINR approximation and the change of variable technique as explained in subsection \ref{sec:noqosdist}. The resulting optimization problem for $N$ TX-RX pairs with QoS constraints can then be written as,
\begin{equation}
\max_{P_{1}^{'},P_{2}^{'},...,P_{N}^{'}} \:
\frac{1}{\ln(2)} \sum_{i=1}^N \ln\left(\frac{g_{ii} e^{P_{i}^{'}}}{\sum_{j \neq i}e^{P_{j}^{'}}g_{ji}+\sigma^2}\right)
\label{new_probq}
\end{equation}
\begin{equation}
\text{subject to:}\: \frac{1}{\ln(2)}  \ln\left(\frac{g_{ii} e^{P_{i}^{'}}}{\sum_{j \neq i}e^{P_{j}^{'}}g_{ji}+\sigma^2}\right) \geq R_i^{min} \:, \forall i
\label{approx_rateq}
\end{equation}
\begin{equation}
\sum_{i=1}^N e^{P_{i}^{'}} \leq P_T
\label{new_constuq}
\end{equation}
We can solve this convex optimization problem using the Lagrange dual decomposition theory. Let $\mu_i, \:i=1,\ldots N$ denote the Lagrange multipliers associated with QoS constraints \eqref{approx_rateq} and $\lambda$ be the Lagrange multiplier associated with constraint (\ref{new_constuq}). We can define the following Lagrangian function,
\begin{align}
 \mathcal{L}_2  & (P_{1}^{'},..,P_{N}^{'},\mu_1,..,\mu_N,\lambda)  =  \frac{1}{\ln(2)} \sum_{i=1}^N \ln\left(\frac{g_{ii} e^{P_{i}^{'}}}{\sum_{j \neq i}e^{P_{j}^{'}}g_{ji}+\sigma^2}\right) \nonumber \\& +  \sum_{i=1}^N \mu_i \left(\frac{1}{\ln(2)} \sum_{i=1}^N \ln\bigg(\frac{g_{ii} e^{P_{i}^{'}}}{\sum_{j \neq i}e^{P_{j}^{'}}g_{ji}+\sigma^2}\bigg) - R_i^{min} \right) \nonumber \\& + \lambda\left(P_T - \sum_{i=1}^N e^{P_{i}^{'}}\right)
\nonumber 
\label{lagrangeq}
\end{align}
Since the problem is convex, Karush-Kuhn-Tucker (KKT) conditions are sufficient to obtain the solution. For any TX-RX pair $i$, we put,
\[\frac{\partial \mathcal{L}_2(P_{1}^{'},P_{2}^{'},...,P_{N}^{'},\lambda)}{\partial P_{i}^{'}} = 0\]
which leads us to,
\begin{equation}
(1+\mu_i) - \sum_{k \ne i}\frac{(1+\mu_k) g_{ik} e^{P_{i}^{'}}}{\sum_{j \neq k}e^{P_{j}^{'}}g_{jk}+\sigma^2} - \lambda \ln(2) e^{P_{i}^{'}} = 0
\label{lagu_1q}
\end{equation}
Now we switch back to the original optimization variables by putting $e^{P_{i}^{'}} = P_i$ in (\ref{lagu_1}) to get,
\begin{align}
P_i & = \frac{(1+\mu_i)} {\lambda \ln(2)  + \sum_{k \ne i} \frac{(1+\mu_k) g_{ik} }{\sum_{j \neq k}P_j g_{jk}+\sigma^2}} \nonumber \\&
= \frac{(1+\mu_i)}{\lambda \ln(2) + \sum_{k \ne i}  \frac{(1+\mu_k) g_{ik} }{\theta_k+\sigma^2}}
\label{P_updateq}
\end{align}
where, $\theta_k$ is as defined in \eqref{theta}.
The values of Lagrange multipliers are iteratively updated according to the sub-gradient method at the start of each iteration $m+1$,
\begin{equation}
\mu_i^{m+1}=\mu_i^m - v(m) \bigg(R_i^{min} - R_i(P_1,P_2,...P_N) \bigg)
\label{sub_updq}
\end{equation}
\begin{equation}
\lambda^{m+1} = \lambda^{m} - v(m)\bigg(P_T - \sum_{i=1}^N P_{i}\bigg)
\label{lambda_updateq}
\end{equation}
where, $v(m)$ denotes the step size used in $m$-th iteration. Based
on these results we develop a distributed Algorithm \ref{algqq}.
This algorithm is implemented individually by each TX $i$. The first
four steps of this algorithm, which enable TX $i$ to compute its
power level, are similar to those developed for Algorithm
\ref{algIV}. In step 5 of this algorithm, the data rate achieved by
TX-RX pair $i$ is computed. In step 6, TX $i$ computes the
difference between its achieved data rate and its minimum target
rate constraint denoted by $\Delta_{R_i}$. The absolute difference
between total available power and sum power utilized by all TXs is
also computed and denoted by $\Delta_{P}$. In step 7, if the
achieved data rate by TX $i$ is greater than $R_i^{min}$, the flag
$S_i^m$ is set to 1. However if the data rate is below the target
data rate constraint, then the flag $S_i^m$ is set to 0, and the
value of $S_i^m$ is then broadcast on the network. In step 10 of
this algorithm, we compare the value of $\Delta_{P}$ with a small
positive constant $\delta>0$. If $\Delta_{P}>\delta$ it indicates
that the sum power constraint is not satisfied. However if
$\Delta_{P}<\delta$, then sum power constraint is satisfied, and the
difference between sum power and available power is less than or
equal to $\delta$. In this step, if sum power constraint is not
satisfied and if the achieved data rate of any TX in the network is
less than its target data rate, the Lagrange multipliers are updated
according to \eqref{sub_updq} and (\ref{lambda_updateq}) and the
algorithm is repeated. Step 11 deals with the case when the sum
power constraint is not satisfied but all the achieved data rates
are greater than the target data rate constraints. In this case,
since sum power constraint is not satisfied, so there are two
possibilities: 1) the TXs are using more sum power than the total
available power to achieve their target data rate constraints, or 2)
they are using less sum power than total available power to achieve
their target data rate constraints. In this case, Lagrange
multipliers are updated and algorithm is repeated so that in the
next iteration achieved data rates are reduced (if sum power was
exceeding $P_T$) or increased (if sum power was less than $P_T$).
Step 12 deals with the case when the sum power constraint is
achieved but target data rate constraints are not satisfied for some
TXs in the network. In this case, it is not possible to achieve all
the target rates with the given amount of power. In other words, the
target data rates are not feasible and the algorithm declares
non-convergence and stops the iterations. Finally, step 13 discuss the
case when all the constraints are satisfied. In this case, the
algorithm declares convergence and stops. When the algorithm
converges, sum power constraint is satisfied with equality (the
difference between sum power and total available power is
$\delta<<1$). However, the achieved data rates are either equal to
greater than the target data rate constraints. In other words,
target data rate constraints are not achieved with strict equality
to allow some TX-RX pairs achieving higher data rates (if there is
enough power to support higher data rates).
\begin{algorithm}[htb]
\caption{Distributed Power Control Algorithm for $N>2$ TX-RX pairs with QoS guarantees:}
\label{algqq}
Initialize $P_i=0, \: \forall i$,  $\mu_i^0=0$, $\lambda^0=0$ and $S_i^0=0,\: \forall i$. \\
Transmit $N-1$ interference channel gain values $g_{ki}$ to other TXs in the network. Receive
$N-1$ interference channel gain values $g_{ik}$ from remaining TXs.
\begin{algorithmic}[1]
\STATE Compute and broadcast the values of $\theta_k$ for other TXs. 
\STATE Receive the values of of $\theta_k$ from remaining TXs ($k \neq i$).
\STATE Calculate $P_i$ using (\ref{P_updateq}) and broadcast it to other TXs.
\STATE Receive the values of of $P_k$ from remaining TXs ($k \neq i$).
\STATE Compute the data rate $R_i(P_1,P_2,...P_N)$ achieved by TX-RX pair $i$.
\STATE Compute $\Delta_{R_i}=R_i(P_1,P_2,...P_N)-R_i^{min}$ and $\Delta_{P}=|P_T - \sum_{i=1}^N P_{i}|$.
\STATE If $\Delta_{R_i} > 0$ set $S_i^m=1$ else set $S_i^m=0$.
\STATE Broadcast the value of $S_i^m$ to other TXs.
\STATE Receive the values of $S_k^m$ from remaining TXs ($k \neq i$).
\STATE If $\Delta_{P}>\delta$ and $S_i^m=0$ for any TX $i$ in the network, update the value of Lagrange multiplier $\mu_i$ using \eqref{sub_updq} and $\lambda$ using (\ref{lambda_updateq}) and go to step 1.
\STATE If $\Delta_{P}>\delta$ and $S_i^m=1,\:\forall i$, update the value of Lagrange multiplier $\mu_i$ using \eqref{sub_updq} and $\lambda$ using (\ref{lambda_updateq}) and go to step 1.
\STATE If $\Delta_{P}<\delta$ and $S_i^m=0$ for any TX $i$ in the network, stop the algorithm and declare non-convergence (in this case data rate constraints cannot be satisfied with given total power $P_T$).
\STATE If $\Delta_{P}<\delta$ and $S_i^m=1,\:\forall i$, stop the algorithm and declare convergence.
\end{algorithmic}
\end{algorithm}

If $\mathcal{M}_2$ denote the number of iterations required for convergence of this algorithm then the signaling overhead $\tilde{S}$ for each TX is, $\tilde{S}=(N-1)+ 3 \mathcal{M}_2$, where $N-1$ interference channel gain information is transmitted at
the start of this algorithm and after that three values ($\theta_i,
P_i, S_i^m$) are transmitted in each iteration. The total signaling
overhead for $N$ TXs is $N \tilde{S}$. Since the power update
equation \eqref{P_updateq} is a standard interference function, and
the problem is approximated to a convex function, the sub-gradient
iterations provide fast convergence to the optimal value of the
Lagrange multiplier with the worst case complexity of $O(\frac{1}{\delta^2})$. 
The step size $v(m)$ in $m+1$-th iteration can again
be obtained according to \eqref{step_rule}.
\subsection{Complexity Analysis of Algorithms with QoS guarantees}
In this subsection, we will do a complexity analysis of Algorithms 5 and 6. 
\subsubsection{Algorithm 5} Algorithm 5 is an analytical power control algorithm for two TX-RX pairs when there are additional QoS constraint. Again this is a non-iterative algorithm, which is developed based on the results of Theorem 5. In this algorithm, the main complexity lies
in solving the quadratic equation when $\tilde{P}>0$. The complexity of this algorithm can therefore be expressed as $O(1)$.
\subsubsection{Algorithm 6} Algorithm 6 is an iterative algorithm for $N$ TX-RX pairs with additional QoS constraints. 
Again the most complex operation in this algorithm is the computation of $N-1$ values of $\theta_k$ with a complexity of $O(N^2)$. Therefore, the worst case complexity of this algorithm is also $O((\frac{1}{\delta})^2 N^2)$ which is polynomial in variables $N$ and $\frac{1}{\delta}$.

\section{Simulation Results}
\label{sec:sim}
We consider a wireless network comprising of multiple TX-RX pairs located inside a circular region of radius 500 m. The TXs are assumed to be uniformly distributed inside the circular region. Each RX is randomly generated and is assumed to be located within a certain radius from its corresponding TX. In the simulations, we consider this radius to be 20 m except in Fig. 6, where we increase this distance to explore the performance over relatively weak direct links. All the TXs are assumed to be operating in the same frequency band. The bandwidth used by each TX is assumed to be 1 MHz. We consider a frequency selective Rayleigh fading channel with exponential delay profile. Path losses among various TXs and RXs are calculated according to Cost-Hata Model \cite{costh}. Power spectral density of noise is assumed to be -114 dBm/Hz. Since TXs and RXs are randomly generated, therefore, channel gains of different links are different i.e. some links have high channel gains while some have low channel gains. It is assumed that all the TXs are connected to a central node and share a single power source.

In the simulations, depending on the number of TX-RX pairs, we compare our results with following different schemes:
\begin{itemize}
\item Water filling power allocation: In this scheme, power is allocated by water-filling over the inverse of direct channel gains. Thus more power gets allocated to TX-RX pairs with higher value of direct channel gain.
\item Pure binary power control: In this scheme, TX with highest value of direct channel gain is allowed to transmit with full power and all the remaining TXs get zero power.
\item Equal power allocation: Power is equally allocated to all the TXs.
\item Exhaustive search: Exhaustive search method can be used to compute optimal power allocation for various TX-RX pairs. We use this scheme to compare the performance of algorithms developed for three or more TX-RX pairs.
\end{itemize}

\subsection{Without QoS Constraints}
We first present the results for the optimization problem without
QoS constraints. In Fig. \ref{fig:fig1} we consider only two TX-RX
pairs. We plot the spectral efficiency (in bits/s/Hz) of the
wireless network versus the total available power. The performance
of all the simulated power allocation schemes increase as the amount
of available power increases. When total power is less, pure binary
power control does not perform as well compared to the proposed
optimal power allocation scheme. The performance of binary power
control improves while the performance of water filling power
allocation scheme degrades as the amount of available power is
increased. Furthermore, we can also observe that the difference in
performance between the proposed optimal algorithm and pure binary
power control decreases as more power is available for allocation.
When the amount of available power is sufficiently large, the
performance of pure binary power control matches the performance of
the optimal algorithm.
\begin{figure}[htb]
\centering
\includegraphics[width=.5\textwidth,height=.22\textheight]{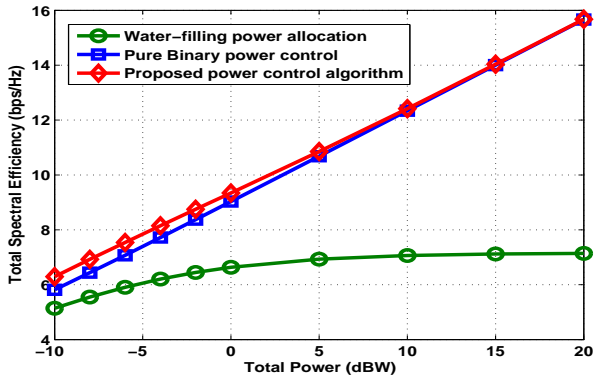}
\caption{Total Spectral Efficiency (bps/Hz) vs Total power (dBW): Two TX-RX pairs without QoS constraints}
\label{fig:fig1}
\end{figure}

In Fig. \ref{fig:fig2}, we plot the total spectral efficiency for
three TX-RX pair problem versus the total available power. We
compare our algorithm with the water-filling power allocation, pure
binary power control and exhaustive search method. The performance
of water-filling power allocation is worst and degrades as available
power increases since it does not consider interference while making
power allocation decisions among the TXs. The performance of pure
binary power control (where the TX with highest value of direct
channel gain transmits with full power and the remaining two TXs
remain silent) is significantly better than the water-filling power
allocation. However, unlike in the two TX-RX pairs case, pure binary
power control does not perform well when more power is available,
and its performance degrades as compared to optimal power
allocation. We can also see from these results that the performance
of our proposed algorithm and that of exhaustive search method is
identical for any amount of available power. However, the complexity
of our algorithm is very low as compared to exhaustive search
method. In our proposed algorithm for three TX-RX pair problem, we
iteratively update the value of power allocated to one TX and find
the power allocation for remaining TXs analytically. In exhaustive
search, all the variables are iteratively updated. The complexity of
our algorithm to find the power allocation for three TX-RX
pairs is thus only a small fraction of that of the exhaustive search
algorithm.
\begin{figure}[htb]
\centering
\includegraphics[width=.5\textwidth,height=.25\textheight]{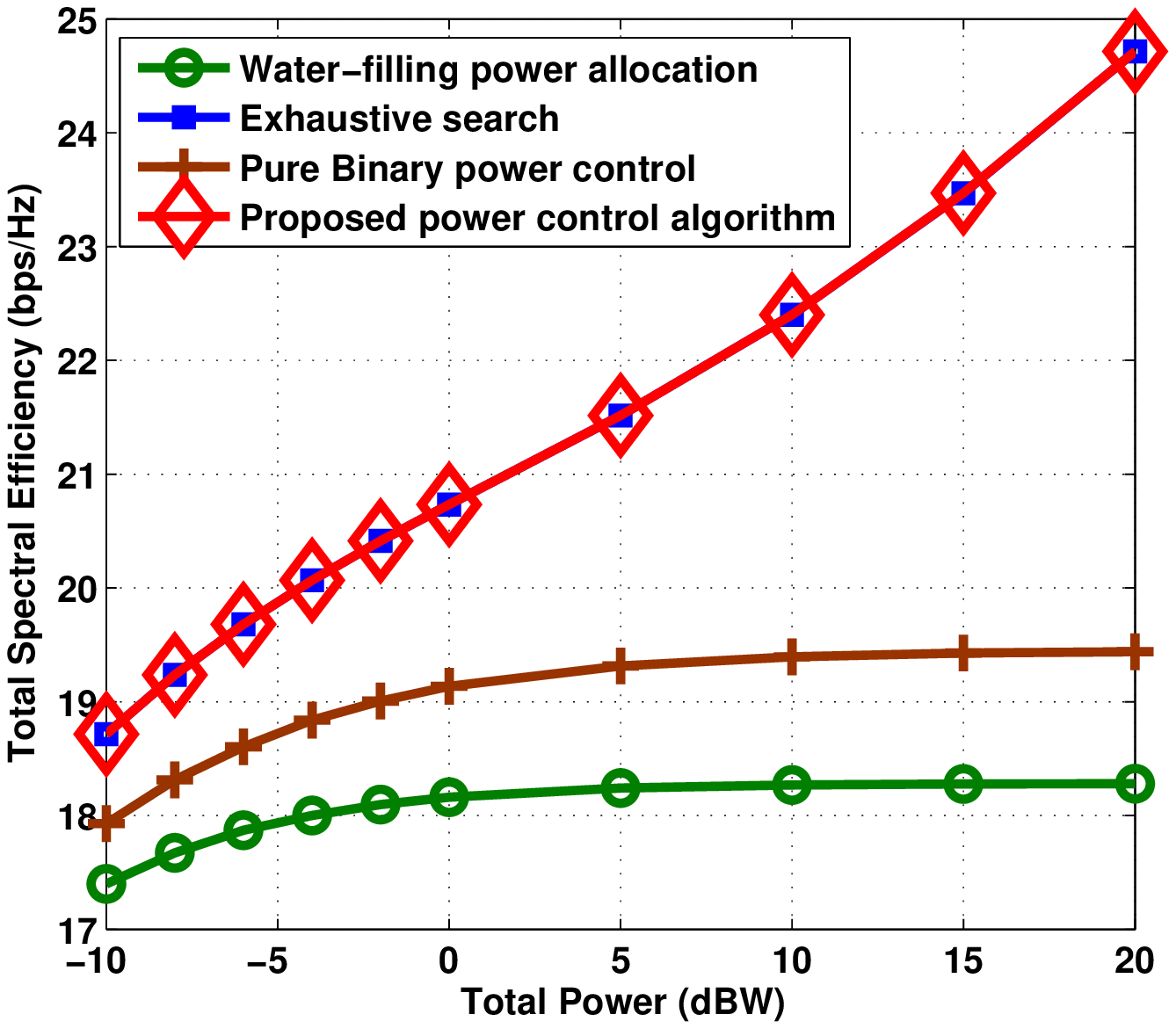}
\caption{Total Spectral Efficiency (bps/Hz) vs Total power (dBs): Three TX-RX pairs without QoS constraints}
\label{fig:fig2}
\end{figure}

In Fig. \ref{fig:fig3}, we plot the total spectral efficiency for
ten TX-RX pair problem vs total available power. We compare our
clustering algorithms (we form clusters of two and three TX-RX
pairs) and distributed algorithm with the exhaustive search method
(which provides the optimal solution). When the total power is less,
the performance of clustering algorithms is relatively better as
compared to the distributed algorithm (which is developed using high
SINR approximation). Forming clusters comprising of three TX-RX
pairs can better manage interference as compared to clusters which
are comprised of only two TX-RX pairs. The performance of all the
sub-optimal algorithms is comparable to that of the optimal solution
provided by the exhaustive search method when available power is
less. As the amount of available power is increased, distributed
algorithm outperforms the clustering algorithms. The performance of
clustering algorithm degrades as total available power is further
increased, since this algorithm cannot manage interference and their
performance approach the performance of equal power allocation
scheme (in equal power allocation scheme power is equally
distributed among all the TXs). 
From these simulations, we can notice that for a generic case of $N$ links, 
we can chose an algorithm based on total available power. When total power is 
less, clustering algorithm may be used and when total power is high distributed 
power control algorithm can be selected. 
However, our proposed distributed power control algorithm performs quite well in most of the scenarios 
under simulation, and it can be implemented in a distributed fashion. One can always improve the 
performance by adaptively selecting algorithms, which serves as an interesting future work.
\begin{figure}[htb]
\centering
\includegraphics[width=.5\textwidth,height=.22\textheight]{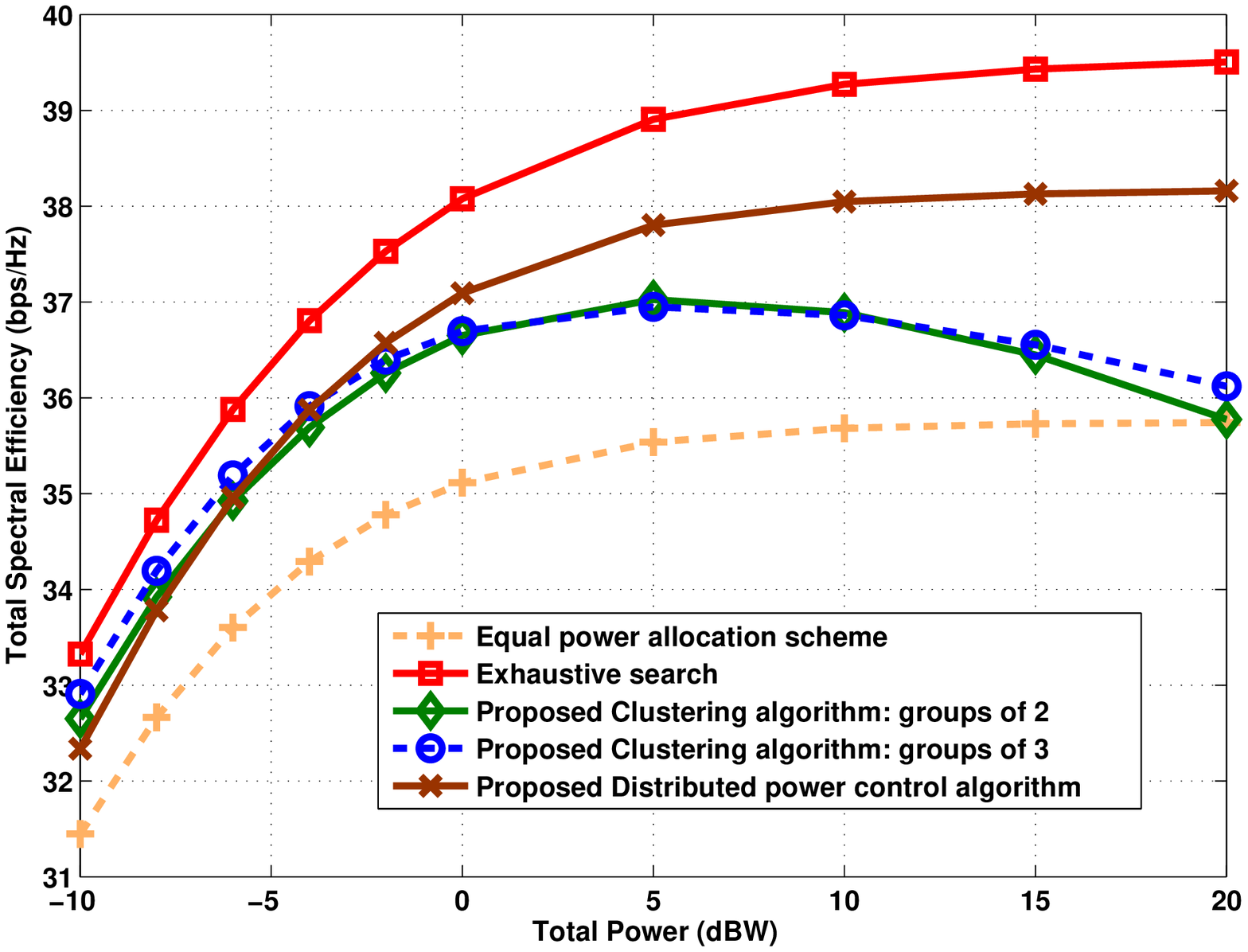}
\caption{Total Spectral Efficiency (bps/Hz) vs Total power (dBs): 10 TX-RX pairs without QoS constraints}
\label{fig:fig3}
\end{figure}

In Fig. \ref{fig:figbar} we assume that total transmit power is 100 Watt (i.e. high total power). 
In these simulations, we consider 6 TX-RX pairs and each RX is randomly generated within a distance 
of 20 m, 50 m and 100 m from its corresponding TX. All other simulation parameters remain same. We compare the
total spectral efficiency of exhaustive search, our proposed distributed power control algorithm and equal power 
allocation scheme. We do not simulate the clustering algorithm since total power is high and we know from
previous simulations that it cannot provide better performance. We can see that as the distance of RX increases
from its corresponding TX, direct channel gains become weak. The performance of equal power allocation scheme
degrades and there is a gap of almost 13 bps/Hz as compared to the exhaustive search algorithm when RX distance is
100 m. On the other hand, the performance of our proposed distributed algorithm is comparable to the exhaustive search
algorithm and well above the equal power allocation scheme. 
\begin{figure}[htb]
\centering
\includegraphics[width=.5\textwidth,height=.22\textheight]{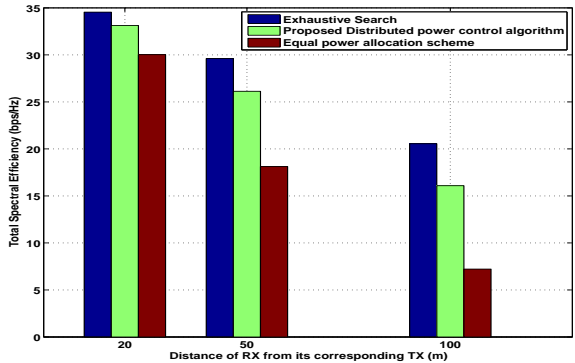}
\caption{Total Spectral Efficiency (bps/Hz) vs Minimum distance of RX from its corresponding TX: Total Transmit power = 100 Watt, 6 TX-RX pairs without QoS constraints}
\label{fig:figbar}
\end{figure}

\subsection{With QoS constraints}
In this subsection, we present the simulation results for the
optimization problem with QoS constraints. In Fig. \ref{fig:figq1}
we plot the total spectral efficiency vs QoS constraints of first
TX-RX pair and second TX-RX pair for different values of total
power. In these figures, if the demanded QoS constraints cannot be
satisfied with given total power or if Algorithm 5 fails then we put total spectral
efficiency equal to zero at all such points. In sub-plot 1, the
total available power is assumed to be -10 dBs and QoS constraint of
first TX-RX pair is varied from 0 to 1 bps/Hz, while QoS constraint
of second TX-RX pair is varied from 0 to 1 bps/Hz. The point where
both the QoS constraints are 0 is highlighted and it corresponds to
the non-QoS case. When there are no QoS constraints, the achieved
sum rate is highest and the achieved rate is the same as the one
obtained using Algorithm \ref{algI} for non-QoS case. As the
demanded data rates by TX-RX pairs increase, the total sum rate
starts to decrease. For a given total power we can identify all 
possible pairs of QoS constraints
$(R_1^{min},R_2^{min})$ which can be simultaneously achieved. For example when total power is -10 dBs, it
is not possible to achieve the QoS constraint pair $(0.8,0.9)$.
However it is possible to achieve a QoS constraint pair $(0.8,0.6)$,
but in order to satisfy these constraints, the sum rate has to be decreased to 1.6 bps/Hz
from 2.894 bps/Hz (sum rate is 2.894 bps/Hz when there are no QoS constraints by
both TX-RX pair 1 and 2, and this point is marked on the graph). As
we increase the total power, achieved sum rate increases, and
more target rate constraint pairs can be achieved. We can see that when the total
power becomes 20 dBs, the achieved sum rate increases to 18.94 bps/Hz
without QoS constraints (as indicated on the plot). At this power
level, QoS constraint pairs as high as $(9,8)$ can also be achieved.
However, again in fulfilling the QoS constraints the achieved sum
rate also decreases. For example to achieve a QoS constraint pair
$(9.6,5.5)$ the sum rate decreases to 15.97 bps/Hz. The decrease in sum
rate is due to the fact that now more power is required to satisfy
QoS constraints which otherwise would have been allocated in some
other way to increase the sum rate.
\begin{figure*}[htb]
\centering
\includegraphics[width=.9\textwidth,height=.5\textheight]{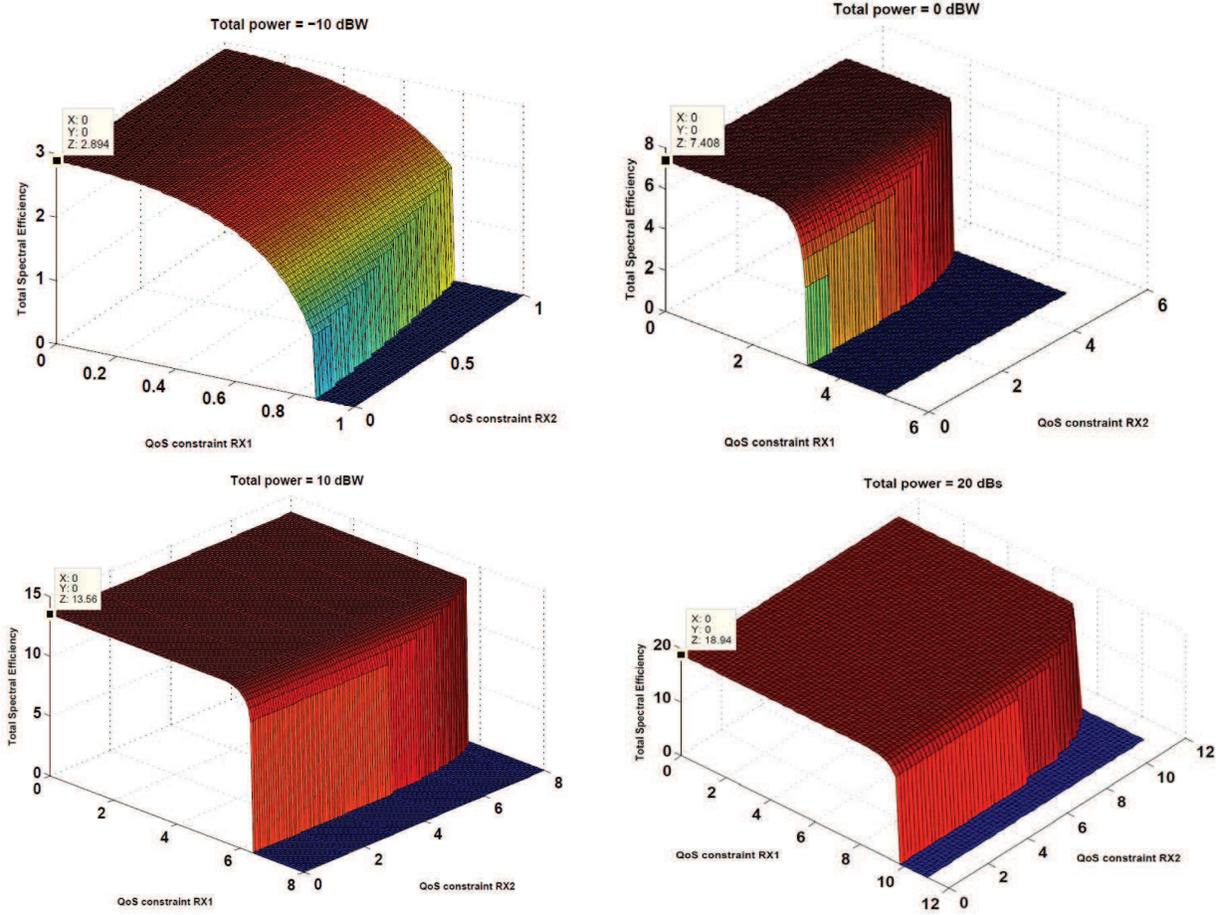}
\caption{Total Spectral Efficiency (bps/Hz) vs QoS constraints of TX-RX pair 1 and 2 for different total power: Two TX-RX pairs}
\label{fig:figq1}
\end{figure*}

In Fig. \ref{fig:figq2}, we present results for ten TX-RX pair
problem with and without QoS constraints. In these simulations, we
assume that 5 TX-RX pairs have a QoS constraint of 2 bps/Hz, while
5 TX-RX pairs have QoS constraint of 1 bps/Hz. We vary the amount of
total power and plot the achieved total spectral efficiency.
Since at lower power levels some QoS constraints cannot be
guaranteed, Algorithm \ref{algqq} declares non-convergence.
We therefore only present the results for power levels where
all the QoS constraints can be satisfied. We can observe that
when there are QoS constraints, the achieved sum rate is less than without QoS
constraint case. Again some power is required to fulfill QoS
constraints which otherwise would have been utilized purely for sum
rate maximization. At very high power values, the
achieved sum rate is identical for with and without QoS constraint
algorithms. In this case since the demanded QoS constraints can be
satisfied with only a small fraction of the total power, remaining
power could purely be utilized for sum rate maximization. 
\begin{figure}[htb]
\centering
\includegraphics[width=.4\textwidth,height=.22\textheight]{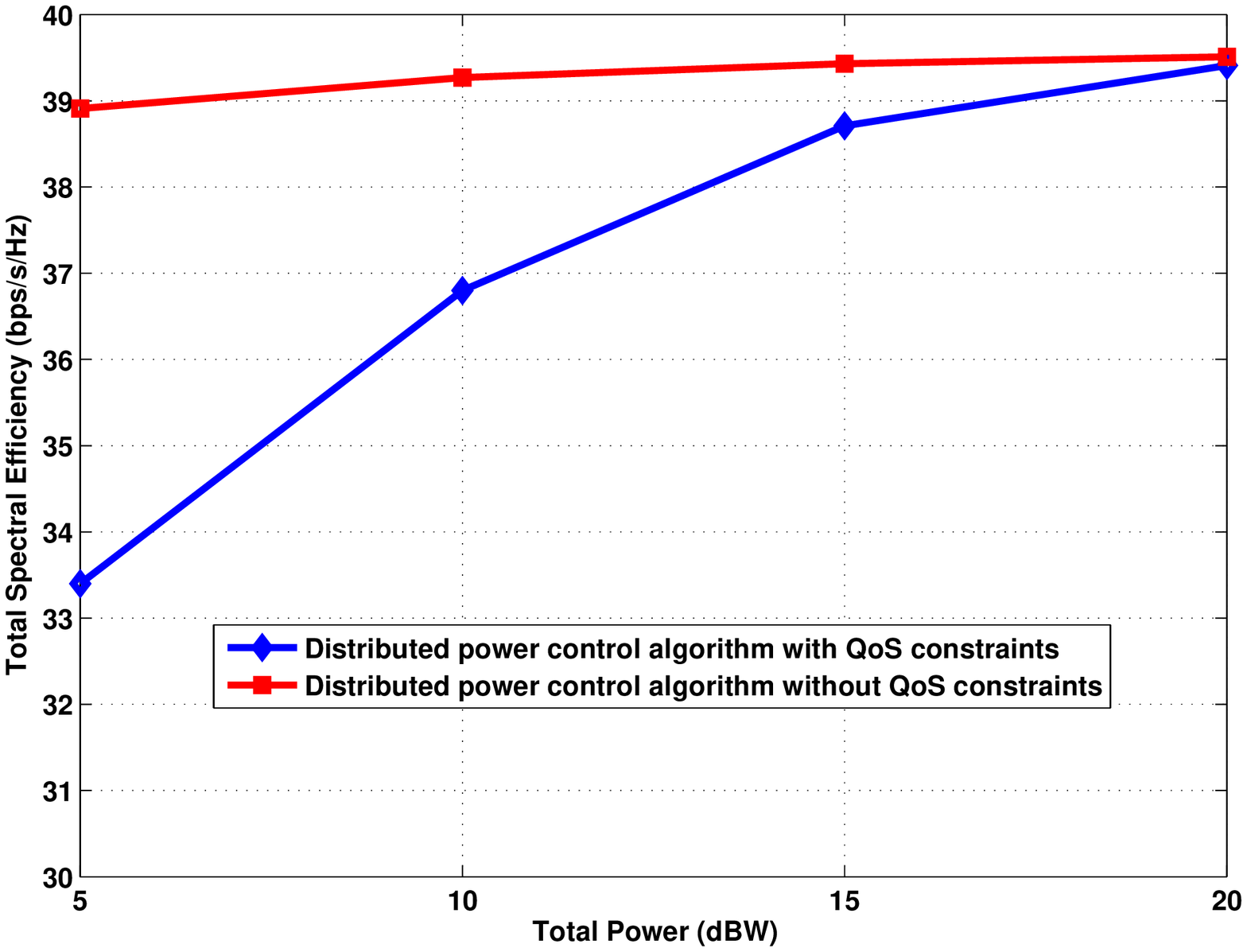}
\caption{Total Spectral Efficiency (bps/Hz) vs Total power (dBs): 10 TX-RX pairs with QoS constraints}
\label{fig:figq2}
\end{figure}

\section{Conclusion}
\label{sec:conc} In this paper, we study the problem of 
power control for sum rate maximization over interference channel
under sum power constraint. We study this problem with and without
individual QoS constraints for each TX-RX pair in the network. When
the objective is only sum rate maximization without QoS guarantees
for individual TXs, we develop an optimal power allocation scheme
for two TX-RX pairs problem and a low complexity iterative algorithm 
for three TX-RX pairs problem. For a generic $N > 3$ TX-RX
pair problem, we develop two low-complexity sub-optimal power
allocation algorithms. The first algorithm forms clusters of two or
three TX-RX pairs and power allocation to TX-RX pairs within each
cluster is then decided by the algorithms developed for two and
three TX-RX pair problems. The second algorithm can be implemented
in a distributed manner and is developed by using a high SINR
approximation. The simulated network scenarios in the paper indicate
that we can use clustering algorithm when total power
is less while distributed algorithm can be used when total power is high. 
We then consider the same problem but with
additional QoS guarantees for individual links. When there are QoS
constraints for individual TX-RX pair, we develop an 
analytical algorithm for two TX-RX pairs problem. For a general case
of $N$ TX-RX pairs, we again develop a distributed low complexity
sub-optimal algorithm. The performance of the developed algorithms is verified using
simulations.

\appendix
\subsection{Proof of Theorem 4}
\label{app_1}
For a given value of $P_1$, the remaining power $\bar{P}=P_T-P_1$ has to be optimally allocated between TX2 and TX3. We can find the optimal values of $(P_2,P_3)$ by taking into consideration the following two cases:
\subsubsection{Case 1}
The remaining power $\bar{P}$ is shared between TX2 and TX3 i.e. $0 < P_i < \bar{P}, \: i=2,3$. In this case, we can put, $P_2=\bar{P}-P_3$ and re-write (\ref{my_eq3}) as a function of $P_3$,
\begin{align*}
 \mathcal{C}(P_3) &= \log_2\bigg(\frac{(a_1\bar{P}+1)+(b_1-a_1)P_3}{b_1P_3+1}\bigg) + \\&  \log_2\bigg(\frac{(c_1\bar{P}+1)+(d_1-c_1)P_3}{(c_1\bar{P}+1)-c_1P_3}\bigg) + \\& \log_2\bigg(\frac{(f_1\bar{P}+e_1+1)+(h_1-f_1)P_3}{(f_1\bar{P}+1)+(h_1-f_1)P_3}\bigg)
\end{align*}
This equation can be compactly written as,
\begin{align}
\mathcal{C}(P_3)= & \log_2\bigg(\frac{a_1^{'}+ b_1^{'}P_3}{1+b_1P_3}\bigg) +  \log_2\bigg(\frac{c_1^{'}+d_1^{'}P_3}{c_1^{'}-c_1P_3}\bigg) + \nonumber \\& \log_2\bigg(\frac{e_1^{'}+h_1^{'}P_3}{f_1^{'}+h_1^{'}P_3}\bigg)
\label{cap_1}
\end{align}
by defining $a_1^{'},b_1^{'},c_1^{'},d_1^{'},e_1^{'},f_1^{'},h_1^{'}$ as in Theorem IV.
We can find the local optima by taking the derivative of (\ref{cap_1}) w.r.t $P_3$ and putting it equal to zero i.e., $\frac{\partial \mathcal{C}(P_3)}{\partial P_3} = 0$; 
which leads us to the following quartic equation.
\[A_1P_{3}^{4}+B_1P_{3}^{3}+C_1P_{3}^{2}+D_1P_{3}+E_1 = 0 \]
where, $A_1,B_1,C_1,D_1,E_1$ are again as defined in Theorem IV. This quartic equation can be solved for $P_3$ by any non-iterative method such as Ferrari's Method. The value of $P_2=\bar{P}-P_3$. However the solutions of this equation are not guaranteed to be feasible or providing global maxima in the interval in which $P_2$ and $P_3$ are defined.
\subsubsection{Case 2} If all the local optima that we find in the previous case are either minima or infeasible, then the maximum of $\mathcal{C}(P_3)$ must lie on either one of the extreme points of the interval in which it is defined. In this case the remaining power is allocated to only of the TXs and the maximum will occur at $(0,\bar{P})$ or $(\bar{P},0)$. \\
Looking at both the cases simultaneously proves the theorem $\blacksquare$

\subsection{Proof of Theorem 5}
\label{app_2}
The proof of this theorem is similar to the proof of Theorem \ref{them_4}.
We can find the values of ($\tilde{P}_1^*,\tilde{P}_2^*$) by taking into consideration the following two cases:
\subsubsection{Case 1}
The remaining power is shared between TX1 and TX2. In this case we can write the objective function \eqref{new_rate} in terms of a single variable $\tilde{P}_1$ by replacing $\tilde{P}_2=\tilde{P}-\tilde{P}_1$ which leads us to,
\begin{align}
&\max_{\tilde{P}_1}\: R_1(\tilde{P}_1)+R_2(\tilde{P}_1) \nonumber \\ &
= \log_2\left(\frac{a_2'+b_2' \tilde{P}_1}{c_2'-b\tilde{P}_1}\right) + \log_2\left(\frac{d_2'+e_2' \tilde{P}_1}{f_2'+c\tilde{P}_1}\right)
\label{new_ratee}
\end{align}
where, $a_2', b_2', c_2', d_2', e_2', f_2'$ are as defined in Theorem \ref{them_33}.
Now putting $\frac{\partial \left(R_1(\tilde{P}_1)+R_2(\tilde{P}_1) \right)}{\partial \tilde{P}_1}=0$ leads us to the following quadratic equation,
\[A_2 \tilde{P}_1^2 + B_2 \tilde{P}_1 +C_2 = 0 \]
where $A_2, B_2, C_2$ are again as defined in Theorem \ref{them_33}. The solutions of this quadratic equation are not guaranteed to be feasible or providing global maxima in the interval $0 < \tilde{P}_1 < \tilde{P}$.
\subsubsection{Case 2} If all the local optima that we find in the previous case are either minima or infeasible, then the maximum of \eqref{new_ratee} must lie on either one of the extreme points of the interval $0 < \tilde{P}_1 < \tilde{P}$. In this case the remaining power is allocated to only of the TXs i.e., $(0,\tilde{P})$ or $(\tilde{P},0)$. \\
Looking at both the cases simultaneously proves the theorem $\blacksquare$

\begin{wrapfigure}{l}{20mm}
  \begin{center}
    \includegraphics[width=.12\textwidth,height=.1\textheight]{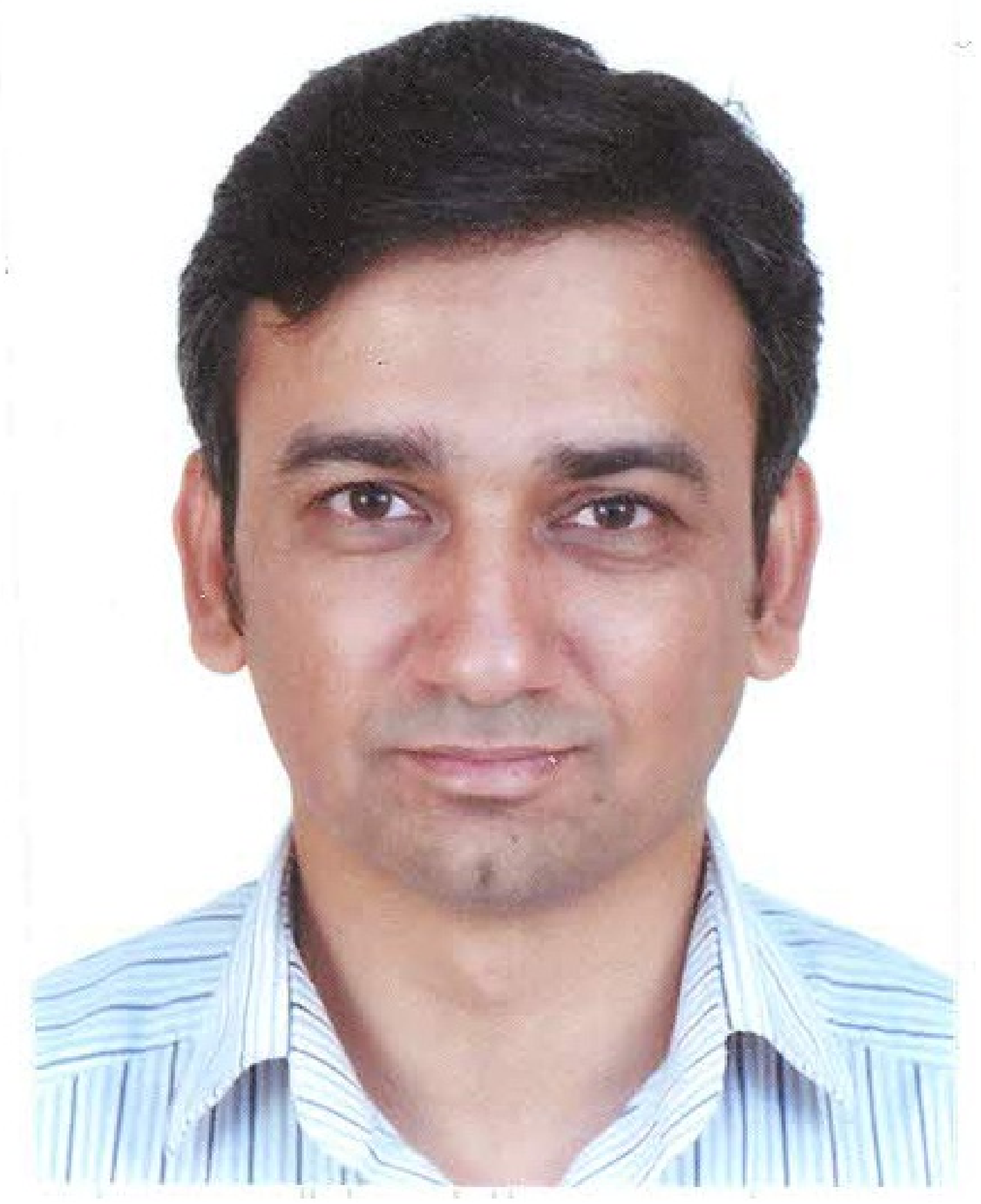}
  \end{center}  
\end{wrapfigure}
\textbf{Naveed UL Hassan} received the B.E degree in Avionics Engineering from College of Aeronautical Engineering, Pakistan in 2002. In 2006 and 2010 he received Masters and PhD degrees in Telecommunications from Ecole Superieure d'Electricite (Supelec) in Gif-sur-Yvette, France. Since 2011, he is an Assistant Professor of Electrical Engineering at Lahore University of Management Sciences (LUMS), Pakistan. He was a visiting Assistant Professor at Singapore University of Technology and Design in 2012 and 2013. His research interests include cross layer design and resource optimization in wireless networks, heterogeneous networks and demand response management in smart grids.

\begin{wrapfigure}{l}{20mm}
  \begin{center}
    \includegraphics[width=.12\textwidth,height=.1\textheight]{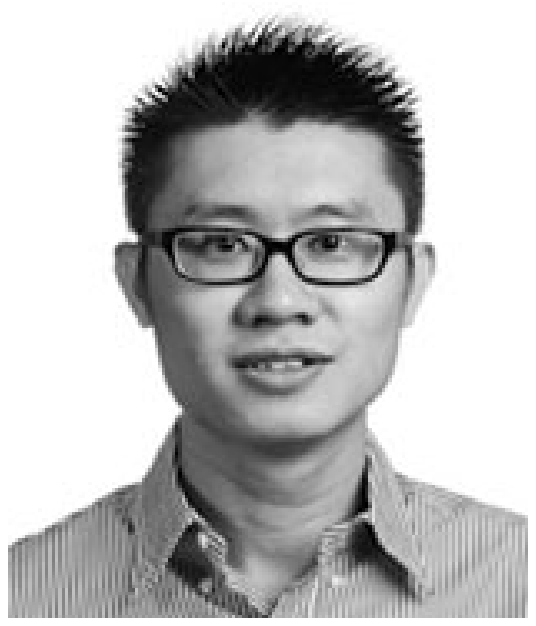}
  \end{center}  
\end{wrapfigure}
\textbf{Chau Yuen} received the BEng and PhD degree from Nanyang Technological University (NTU), Singapore, in 2000 and 2004 respectively. He is the recipient of Lee Kuan Yew Gold Medal, Institution of Electrical Engineers Book Prize, Institute of Engineering of Singapore Gold Medal, Merck Sharp \& Dohme Gold Medal and twice the recipient of Hewlett Packard Prize. Dr Yuen was a Post Doc Fellow in Lucent Technologies Bell Labs, Murray Hill during 2005. He was a Visiting Assistant Professor of Hong Kong Polytechnic University in 2008. During the period of 2006 ‐ 2010, he worked at the Institute for Infocomm Research (I2R, Singapore) as a Senior Research Engineer, where he was involved in an industrial project on developing an 802.11n Wireless LAN system, and participated actively in 3Gpp Long Term Evolution (LTE) and LTE‐Advanced (LTE‐A) standardization. He joined the Singapore University of Technology and Design as an assistant professor from June 2010, and received IEEE Asia-Pacific Outstanding Young Researcher Award on 2012. Dr Yuen serves as an Associate Editor for IEEE Transactions on Vehicular Technology.

\begin{wrapfigure}{l}{20mm}
  \begin{center}
    \includegraphics[width=.12\textwidth,height=.1\textheight]{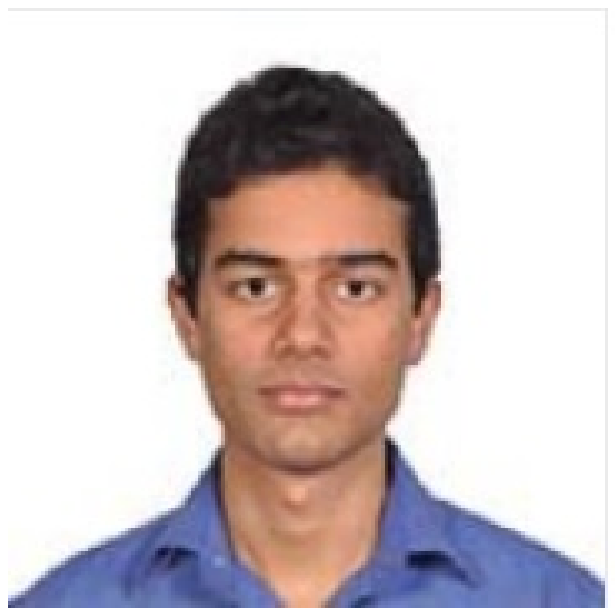}
  \end{center}  
\end{wrapfigure}
\textbf{Shayan Saeed} received the BS degree in Electrical Engineering from Lahore University of Management Sciences in 2013. Currently he is working towards his Ph.D. degree in computer science at University of Illinois, Urbana-Champaign. His industry experience includes a research internship with the AT\&T labs. His research interests span areas of storage systems, communication networks and computer systems.

\begin{wrapfigure}{l}{20mm}
  \begin{center}
    \includegraphics[width=.12\textwidth,height=.1\textheight]{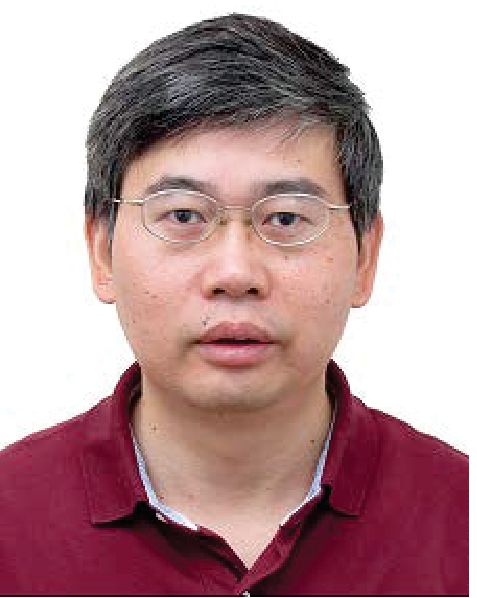}
  \end{center}  
\end{wrapfigure}
\textbf{Zhaoyang Zhang} (M’02) received his B.S. and Ph.D degrees in communication and information systems from Zhejiang University, China, in 1994 and 1998 respectively. He is currently a full professor with the Department of Information Science and Electronic Engineering, Zhejiang University. His research interests are mainly focused on information theory, signal processing and their applications in wireless communications and networking. He has published more than 150 refereed international journal and conference papers and two books in the above areas. He was a co-recipient of three conference Best Paper Awards or Best Student Paper Award. He served or is serving as Associate Editor for international journals like IET Communcations, Wiley International Journal of Communication Systems, Elsevier Physical Communications, etc., and TPC Co-Chair for WCSP’ 2013, ICUFN’2011/12/13, etc., and Co-chair for Globecom 2014 Wireless Communications Symposium, etc.

\end{document}